\documentclass[11pt]{article}

\usepackage{amssymb}
\usepackage{amsfonts}
\usepackage{amsmath}
\usepackage{amsthm}
\usepackage[english]{babel}
\usepackage{latexsym}
\usepackage{color}
\usepackage{enumerate}
\usepackage{enumitem}
\usepackage{nicefrac}
\usepackage{mathrsfs}
\usepackage{comment}
\usepackage[hmargin=2cm,vmargin=1.9cm]{geometry}
\usepackage{bm}
\usepackage{bbm}

\usepackage{centernot}

\usepackage[affil-it]{authblk}
\theoremstyle{plain}
\newtheorem{theorem}{Theorem}[section]
\newtheorem{lemma}[theorem]{Lemma}
\newtheorem{example}[theorem]{Example}
\newtheorem{proposition}[theorem]{Proposition}
\newtheorem{corollary}[theorem]{Corollary}

\theoremstyle{definition}
\newtheorem{definition}[theorem]{Definition}
\newtheorem{assumption}[theorem]{Assumption}
\newtheorem*{assumption*}{Standing Assumption}
\newtheorem{remark}[theorem]{Remark}

\theoremstyle{remark}

\newtheoremstyle{claim}
  {10pt}   
  {10pt}   
  {\itshape}  
  {0pt}       
  {\bfseries} 
  {.}         
  {5pt plus 1pt minus 1pt} 
  {}          

\theoremstyle{claim}

\numberwithin{equation}{section}

\newcommand{\ba}{\begin{array}{ll}}
\newcommand{\bal}{\begin{array}{ll}}
\newcommand{\ea}{\end{array}}

\newcommand{\E}{\mathbb{E}}

\newcommand{\probp}{\mathbb{P}}
\newcommand{\probq}{\mathbb{Q}}
\newcommand{\R}{\mathbb{R}}
\newcommand{\N}{\mathbb{N}}

\newcommand{\cN}{{\mathcal{N}}}
\newcommand{\cF}{{\mathcal{F}}}

\newcommand{\cG}{{\mathcal{G}}}
\newcommand{\cB}{{\mathcal{B}}}
\newcommand{\cS}{{\mathcal{S}}}
\newcommand{\cA}{\mathcal{A}}
\newcommand{\cC}{\mathcal{C}}
\newcommand{\cD}{\mathcal{D}}

\newcommand{\cK}{\mathcal{K}}

\newcommand{\cL}{\mathcal{L}}
\newcommand{\cM}{\mathcal{M}}

\newcommand{\cP}{\mathcal{P}}

\newcommand{\cX}{{\mathcal{X}}}
\newcommand{\cY}{{\mathcal{Y}}}

\newcommand{\one}{\mathbbm 1}

\newcommand{\barr}{\mathop{\rm bar}\nolimits}

\newcommand{\cone}{\mathop{\rm cone}\nolimits}
\newcommand{\co}{\mathop{\rm co}\nolimits}

\newcommand{\Span}{\mathop{\rm span}\nolimits}
\newcommand{\cl}{\mathop{\rm cl}\nolimits}

\newcommand{\VaR}{\mathop {\rm VaR}\nolimits}

\newcommand{\ES}{\mathop {\rm ES}\nolimits}

\newcommand{\MCP}{\mathop {\rm MCP}\nolimits}

\def\keywords{\vspace{.5em}
{\noindent\textbf{Keywords}:\,\relax%
}}

\def\JELclassification{\vspace{.5em}
{\noindent\textbf{JEL classification}:\,\relax%
}}

\def\MSCclassification{\vspace{.5em}
{\noindent\textbf{MSC}:\,\relax%
}}

\makeatletter
\def\@fnsymbol#1{\ensuremath{\ifcase#1\or 1\or 2\or 3\or 4\or 5\or 6\or 7\or 8\else\@ctrerr\fi}}
\makeatother

\makeatletter
\renewcommand\@biblabel[1]{}
\makeatother

\begin{document}

\title{Fundamental theorem of asset pricing with acceptable risk\\ in markets with frictions}

\author{\sc{Maria Arduca}}
\affil{Department of Economics and Finance\\
LUISS University, Rome, Italy\\
\texttt{\normalsize marduca@luiss.it}}
\author{\sc{Cosimo Munari}}
\affil{Center for Finance and Insurance and Swiss Finance Institute\\
University of Zurich, Switzerland\\
\texttt{\normalsize cosimo.munari@bf.uzh.ch}}

\date{\today}

\maketitle

\begin{abstract}
We study the range of prices at which a rational agent should contemplate transacting a financial contract outside a given securities market. Trading is subject to nonproportional transaction costs and portfolio constraints and full replication by way of market instruments is not always possible. Rationality is defined in terms of consistency with market prices and acceptable risk thresholds. We obtain a direct and a dual description of market-consistent prices with acceptable risk. The dual characterization requires an appropriate extension of the classical Fundamental Theorem of Asset Pricing where the role of arbitrage opportunities is played by good deals, i.e., costless investment opportunities with acceptable risk-reward tradeoff. In particular, we highlight the importance of scalable good deals, i.e., investment opportunities that are good deals regardless of their volume.
\end{abstract}


\JELclassification{D81, G12}

\MSCclassification{91B25, 91G20, 91G30, 91G70}

\parindent 0em \noindent


\section{Introduction}


One of the fundamental goals of financial economics is to investigate at which price(s) a rational agent should contemplate transacting a financial contract. The pricing problem can be addressed in a number of ways depending on the assumptions on the agent as well as on the contract considered. The point of departure of classical arbitrage pricing is the assumption that agents are wealth maximizers and have access to an outstanding market where a number of basic financial securities are traded for a known price in an arbitrage-free way. The task is to find at which prices an agent would be willing to transact a given financial contract outside of the market. As is well known, the corresponding range of rational prices coincides with the interval of arbitrage-free prices. Since the pioneering contributions of Black and Scholes (1973), Merton (1973), Cox and Ross (1976), Rubinstein (1976), Ross (1978), Harrison and Kreps (1979), Kreps (1981), this framework has successfully been extended in several directions. A prominent line of research has contributed to what may be broadly called a general theory of ``subjective pricing''. This has been achieved by investigating the pricing problem under suitable relaxations of the classical notion of an arbitrage opportunity. A key contribution in this direction is the theory of good deal pricing initiated by Cochrane and Saa Requejo (2000) and Bernardo and Ledoit (2000) and based on the idea of restricting the interval of arbitrage-free prices by incorporating individual ``preferences'' into the pricing problem. This leads to tighter pricing bounds called good deal bounds. In this setting, arbitrage opportunities are replaced by good deals, i.e., investment opportunities that require no funding costs and deliver terminal payoffs that are sufficiently attractive based on the agent's ``preferences''. The crucial point is that, differently from arbitrage opportunities, good deals may expose to downside risk and the agent's task is therefore that of determining acceptable risk thresholds. Several ways to define 
risk thresholds have been considered in the literature, e.g., by means of Sharpe ratios in Cochrane and Saa Requejo (2000), Bj\"{o}rk and Slinko (2006), and Bion-Nadal and Di Nunno (2013), gain-loss ratios in Bernardo and Ledoit (2000), test probabilities in Carr et al.\ (2001), utility functions in \v{C}ern\'{y} and Hodges (2002), \v{C}ern\'{y} (2003), Kl\"{o}ppel and Schweizer (2007), and Arai (2011), expected shortfall in Cherny (2008), distance functions in Bondarenko and Longarela (2009), and acceptability indices in Madan and Cherny (2010). A theory for general acceptance sets has been developed by Jaschke and K\"{u}chler (2001), \v{C}ern\'{y} and Hodges (2002), Staum (2004), Cherny (2008), and Cheridito et al.\ (2017). We also refer to Arai and Fukasawa (2014) and Arai (2017) for a study of optimal good deal pricing bounds. One can distinguish between two research directions in the field. A first strand of literature starts by imposing suitable constraints on price deflators or, equivalently, martingale measures with the aim of  restricting the interval of arbitrage-free prices. The resulting good deal bounds can be therefore expressed in dual terms. The rationale for discarding some arbitrage-free prices is that transacting at those prices would create good deals with respect to a suitable acceptance set. The task is precisely to characterize the corresponding acceptance set. A second strand of literature starts by tightening the superreplication price through a suitable enlargement of the cone of positive random variables, which is replaced by a larger acceptance set. The task is to establish a dual description of the resulting good deal bounds. This is achieved by extending the Fundamental Theorem of Asset Pricing to a good deal pricing setting.

\smallskip

In this paper we follow the second strand of research mentioned above. Our goal is to contribute to the literature on good deal pricing in a static setting by establishing a version of the Fundamental Theorem of Asset Pricing in incomplete markets with frictions where agents use general acceptance sets to define good deals based on their individual ``preferences''. The presence of general acceptance sets poses technical challenges and requires pursuing a new strategy as the standard change of numeraire and exhaustion arguments behind the classical proof of the Fundamental Theorem can no longer be exploited. The highlights of our contribution are the following:
\begin{itemize}[leftmargin=*]
  \item The point of departure is a clear and economically motivated definition of rational prices that is missing in the good deal pricing literature with the exception of Cherny (2008). Our approach is different and inspired by Koch-Medina and Munari (2020). We assume that an agent willing to purchase a financial contract outside of the market will never accept to buy at a price at which he or she could find a better replicable payoff in the market. In the spirit of good deal pricing, the agent is prepared to accept a suitable ``replication error'', which is formally captured by an acceptance set. The corresponding rational prices are called market-consistent prices. In a frictionless setting where agents accept no ``replication error'' our notion boils down to the classical notion of an arbitrage-free price.
  \item We work under general convex transaction costs and portfolio constraints, which allows us to model both proportional and nonproportional frictions. The bulk of the literature has focused on frictionless markets or markets with proportional trasaction costs. Portfolio constraints have been rarely considered. Moreover, instead of focusing on the set of attainable payoffs at zero cost as a whole, we state our results by explicitly highlighting the specific role played by each source of frictions.
  \item The payoff space is taken to be the space of random variables over a general probability space. This is different from the bulk of the literature, with the exception of Cherny (2008), where regularity conditions on payoffs, e.g., integrability, are stipulated upfront in view of the application of special mathematical results, e.g., duality theory. The advantage of our approach is that we are able to highlight where and why a restriction to a special class of payoffs is needed, e.g., to apply duality theory, and what are its consequences in terms of the original pricing problem. This also allows us to point out the failure of change of numeraire techniques applied to acceptance sets. This important aspect, which distinguishes good deal pricing from arbitrage pricing, has never been discussed in the literature.
  \item  We introduce the notion of scalable good deals, i.e., payoffs that are good deals independently of their size, which extends to a good deal pricing setting the notion of a scalable arbitrage opportunity by Pennanen (2011a). The absence of scalable good deals is key to deriving our characterizations of market-consistent prices. This condition is weaker than the absence of good deals commonly stipulated in the literature. In particular, there are situations where absence of arbitrage is sufficient to ensure absence of scalable good deals. We also argue that absence of scalable good deals is economically sounder than absence of good deals.
  \item We adapt the classical notion of a price deflator to our good deal setting with frictions and introduce the class of strictly-consistent price deflators, which correspond to the Riesz densities of a pricing rule in a complete frictionless market where the basic traded securities are ``priced'' in accordance with their (suitably adjusted in the presence of nonproportional frictions) bid-ask spreads and every nonzero acceptable payoff has a strictly positive ``price''. This is different from similar notions in the literature, where no bid-ask spread adjustments are considered and acceptable payoffs, including positive payoffs, are often assumed to have a nonnegative ``price'' only.
  \item We establish direct and dual characterizations of market-consistent prices. The direct characterization is based on the analysis of superreplication prices and extends to a good deal pricing setting the classical findings of Bensaid et al.\ (1992) in markets with frictions. The dual characterization is based on a general version of the Fundamental Theorem of Asset Pricing and underpins the appropriate extension of the classical superreplication duality. Under suitable assumptions on the underlying model space, the Fundamental Theorem establishes equivalence between absence of scalable good deals and existence of strictly-consistent price deflators. This extends to a good deal pricing setting the static version of the Fundamental Theorem obtained by Pennanen (2011a). We provide a detailed comparison with the literature to highlight in which sense our result extends and sharpens the various formulations of the Fundamental Theorem in the good deal pricing literature. The only work on good deal pricing featuring a strong result with strictly-consistent price deflators is \v{C}ern\'{y} and Hodges (2002). In that paper the market is frictionless and the acceptance set is assumed to be boundedly generated, a condition that often forces the underlying probability space to be finite.
\end{itemize}
The paper is organized as follows. In Section~\ref{sect: market model} we describe the market model and the agent's acceptance set, and we introduce the notion of market-consistent prices with acceptable risk. In Section~\ref{sect: acceptable deals} we focus on good deals and show a number of sufficient conditions for the absence of scalable good deals (Proposition~\ref{prop: no acc deals}). Our main results are recorded in Section~\ref{sect: FTAP}. We establish a direct and a dual characterization of market-consistent prices with acceptable risk (Propositions~\ref{theo: characterization mcp superreplication}
and~\ref{theo: dual MCP}). The dual characterization is based on our general version of the Fundamental Theorem of Asset Pricing (Theorem~\ref{theo: FTAP})
and the corresponding Superreplication Theorem (Theorem~\ref{theo: superhedging theorem}), which are, from a technical perspective, the highlights of the paper. Throughout we prove sharpness of our results by means of suitable examples, which are always presented in the simplest possible setting, namely that of a two-states model, to demonstrate their general validity.

\section{The pricing problem}
\label{sect: market model}

In this section we state the pricing problem and describe the underlying mathematical framework. The bulk of the presentation is aligned with our reference literature on good deal pricing, e.g., Carr et al.\ (2001), Jaschke and K\"{u}chler (2001), \v{C}ern\'{y} and Hodges (2002), Staum (2004), Cherny (2008), Madan and Cherny (2010). We will highlight discrepancies where needed.

\smallskip

We consider an agent who has access to a financial market where a finite number of basic securities are traded. The agent's problem is to determine the range of prices at which he or she should contemplate transacting a financial contract outside of the market. The candidate prices should satisfy the following rationality conditions. On the one hand, they should be {\em consistent with the market}, i.e., the agent should not be willing to transact if the market offers a better contract for a better price. On the other hand, they should be {\em consistent with individual ``preferences''}, i.e., the agent should determine when a marketed contract is better based on a pre-specified criterion of acceptability. To define market-consistent prices with acceptable risk we thus have to describe the underlying market model and the agent's acceptance set. From now on, we always take a buyer's perspective and we therefore focus on ask prices. The conversion to a seller's perspective and to bid prices is straightforward.


\subsection{The market model}
\label{sect: market model details}

We consider a one-period financial market where uncertainty about the terminal state of the economy is captured by a probability space $(\Omega,\cF,\probp)$. We denote by $L^0$ the space of random variables modulo almost-sure equality under $\probp$ and equip it with its canonical algebraic operations and partial order. The set of positive random variables is denoted by $L^0_+$ and is referred to as the standard positive cone. Similarly, for $\cL\subset L^0$ we define $\cL_+:=\cL\cap L^0_+$. The expectation under $\probp$ is denoted by $\E$. For every $X\in L^0$ we define $\E_\probp[X] := \E_\probp[X^+]-\E_\probp[X^-]$, where $X^+$ and $X^-$ are the positive and negative part of $X$, and we follow the sign convention $\infty-\infty=-\infty$. The standard Lebesgue spaces are denoted by $L^p$ for $p\in[1,\infty]$. The elements of $L^0$ represent {\em payoffs} of financial contracts at the terminal date. We identify the elements of $\R$ with constant payoffs and refer to them as {\em risk-free payoffs}.

\smallskip

We assume that a finite number of basic securities are traded in the market and denote by $\cS\subset L^0$ the vector space spanned by their payoffs. The elements of $\cS$ are called {\em replicable payoffs}. Contrary to most of the good deal pricing literature, we do not assume the existence of risk-free replicable payoffs. To each replicable payoff we associate an ask price via a {\em pricing rule} $\pi:\cS\to(-\infty,\infty]$. In line with the literature, we allow for nonfinite prices to account for the existence of physical limitations in the availability of replicable payoffs. These limitations affect every agent. Moreover, we fix a nonempty set $\cM\subset\cS$ consisting of those replicable payoffs that can be effectively bought by our agent. The elements of $\cM$ are called {\em attainable payoffs} and account for the existence of, e.g., regulatory limitations in the purchase of replicable payoffs. These limitations are specific to our agent. Even though the agent has access to a (possibly strict) subset of $\cS$ only, it is mathematically convenient to define $\pi$ on the entire set $\cS$ to exploit its natural vector space structure. Recall that, by finite dimensionality, every linear Hausdorff topology on $\cS$ is normable and any two norms on $\cS$ are equivalent. We stipulate the following assumptions on the market primitives.

\begin{assumption}
We denote by $\|\cdot\|$ a fixed norm on $\cS$. We assume that $\pi$ is convex, lower semicontinuous, and satisfies $\pi(0)=0$. Moreover, we assume that $\cM$ is convex, closed, and satisfies $0\in\cM$.
\end{assumption}

Our general setting is compatible with a variety of market models encountered in the literature.

\begin{example}
\label{ex: market models}
Let $S_1,\dots,S_N\in L^0$ be the payoffs of the basic securities. To avoid redundant securities, assume that they are linearly independent. Through their trading activity, agents can set up portfolios of basic securities at the initial date. A portfolio of basic securities is represented by a vector $x=(x_1,\dots,x_N)\in\R^N$. We adopt the standard convention according to which a positive entry refers to a long position and a negative entry to a short position. Since in our setting no trading occurs at the terminal date and each security delivers its terminal state-contingent contractual payoff, portfolio $x$ generates the payoff $\sum_{i=1}^Nx_iS_i$, and the set of replicable payoffs $\cS$ coincides with the linear space generated by $S_1,\dots,S_N$. To each portfolio we associate an ask price via $V_0:\R^N\to(-\infty,\infty]$. As no basic security is redundant, two portfolios generating the same payoff must coincide and, hence, command the same ask price. This ``law of one price'' allows us to define for every replicable payoff $X\in\cS$
\[
\pi(X) = V_0(x)
\]
where $x\in\R^N$ is any portfolio satisfying $X=\sum_{i=1}^Nx_iS_i$. The pricing rule $\pi$ satisfies the stipulated assumptions whenever $V_0$ is convex, lower semicontinuous, and satisfies $V_0(0)=0$. This is the case in any of the following situations.
\begin{itemize}[leftmargin=*]
  \item {\em No transaction costs}. In a frictionless market the bid-ask spread associated with every basic security is zero so that every unit of the $i$th basic security can be bought or sold for the same price $p_i\in\R$. This yields the classical linear pricing functional
\[
V_0(x) = \sum_{i=1}^Np_ix_i.
\]
  \item {\em Proportional transaction costs}. In a market with proportional transaction costs every unit of the $i$th basic security can be bought for the price $p^b_i\in\R$ and sold for the price $p^s_i\in\R$. It is natural to assume that $p^b_i\geq p^s_i$ so that the corresponding bid-ask spread is nonnegative. In this setting, it is natural to consider the sublinear pricing functional used, e.g., in Jouini and Kallal (1995)
\[
V_0(x) = \sum_{x_i\geq0}p^b_ix_i+\sum_{x_i<0}p^s_ix_i.
\]
  \item {\em Nonproportional transaction costs}. In a market with nonproportional transaction costs the unitary buying and selling prices for the $i$th basic security vary with the volume traded according to some functions $p^b_i,p^s_i:\R_+\to\R\cup\{\infty\}$. Again, it makes sense to assume that $p^b_i(x)\geq p^s_i(x)$ for every $x\in\R_+$ so that the corresponding bid-ask spread is nonnegative. In many market models, see, e.g., the careful discussion about limit-order markets in Pennanen (2011a), it is natural to assume that $p^b_i$ is convex and $p^s_i$ is concave and that both are null and right continuous at zero as well as left continuous at the point where they jump at infinity. In addition, their one-sided derivatives should satisfy $\partial^+p^b_i(0)\geq\partial^+p^s_i(0)$. The assumption that $p^b_i$ and $p^s_i$ take nonfinite values represents a cap on the total number of units available in the market. In this setting, it is natural to consider the convex pricing functional used, e.g., in \c{C}etin and Rogers (2007)
\[
V_0(x) = \sum_{x_i\geq0}p^b_i(x_i)-\sum_{x_i<0}p^s_i(-x_i).
\]
\item {\em General convex pricing functional}. By standard convex duality, all the preceding examples are special instances of the general convex pricing functional defined by
\[
V_0(x) = \sup_{p\in\R^N}\left\{\sum_{i=1}^Np_ix_i-\delta(p)\right\},
\]
where $\delta:\R^N\to[0,\infty]$ is a map attaining the value zero. The map $\delta$ can be used to generate pre-specified deviations from frictionless prices. In particular, differently from the previous rules, this general pricing rule allows for a nonadditive structure across the different basic securities. We refer to Kaval and Molchanov (2006) and Pennanen (2011a) for concrete examples in the setting of link-saved trading and limit-order markets.
\end{itemize}
We model portfolio constraints such as borrowing and short selling restrictions on specific basic securities by restricting the set of admissible portfolios to a subset $\cP\subset\R^N$. The set $\cM$ thus corresponds to
\[
\cM = \left\{\sum_{i=1}^Nx_iS_i \,; \ x\in\cP\right\}.
\]
The set $\cM$ satisfies the stipulated assumptions whenever $\cP$ is convex, closed, and satisfies $0\in\cP$. This is the case in any of the following situations. We refer to Pennanen (2011a) and the references therein for additional examples of portfolio constraints that are compatible with our setting.
\begin{itemize}[leftmargin=*]
  \item {\em No portfolio constraints}. This corresponds to $\cP=\R^N$.
  \item {\em No short selling}. This corresponds to $\cP=\R^N_+$.
  \item {\em Caps on short and long positions}. This corresponds to $\cP=[\underline{x}_1,\overline{x}_1]\times\cdots\times[\underline{x}_N,\overline{x}_N]$ for suitable $\underline{x},\overline{x}\in\R^N$ such that $\underline{x}_i\leq\overline{x}_i$ for every $i=1,\dots,N$. In particular, this allows us to impose no short selling and caps on long positions at the same time.
\end{itemize}
\end{example}

\subsection{The acceptance set}

As said, the agent's problem is to determine a range of prices at which he or she should be prepared to acquire a financial contract with payoff $X\in L^0$ outside of the market. To tackle this problem, the agent will identify among all attainable payoffs $Z\in\cM$ those that are ``preferable'' to $X$ (from a buyer's perspective) and use the corresponding market prices to determine an upper bound on the candidate prices for $X$. In line with good deal pricing theory, we define said ``preference'' relationship by means of an acceptance set $\cA\subset L^0$. More precisely, we assume that $Z$ is ``preferable'' to $X$ whenever
$Z-X\in\cA$. It should be noted that the relation induced by $\cA$ is not a preference relation in a technical sense unless $\cA$ is a convex cone. The bulk of the good deal pricing literature has focused on this special case. This is, however, unsatisfactory as there exist relevant acceptance sets that are convex but fail to be conic, e.g., acceptance sets defined through utility functions or stochastic dominance. To include these examples, we join \v{C}ern\'{y} and Hodges (2002) and Staum (2004) and dispense with conicity. In this case, we find necessary to consequently dispense with the language of ``preferences'' and to provide a new, more general, interpretation to the acceptance set. In this paper, we interpret $\cA$ as the set of all replication errors that are deemed acceptable by the agent. In other words, the agent will try to replicate $X$ by means of attainable payoffs $Z$ available in the market and will use the acceptance set to determine when the residual payoff $Z-X$ is acceptable or not. If $\cA=L^0_+$, the agent will target perfect superreplication, thereby accepting no downside risk in the replication procedure. This choice corresponds to the classical setting of arbitrage pricing. If $\cA$ is strictly larger than $L^0_+$, the agent will be prepared to accept a suitable amount of downside risk. This may be achieved, e.g., by setting a cap on the downside risk alone or by balancing upside and downside risk.

\smallskip

The elements of $\cA$ are called {\em acceptable payoffs}. We assume that every payoff dominating an acceptable payoff is also acceptable and that the notion of acceptability is well behaved with respect to aggregation in the sense that every convex combination of acceptable payoffs remains acceptable. The first property corresponds to the usual monotonicity requirement stipulated in risk measure theory; see, e.g., Artzner et al.\ (1999). Formally, we stipulate the following assumptions
on the acceptance set.

\begin{assumption}
\label{ass: acceptance set}
The set $\cA$ is a strict convex subset of $L^0$ and satisfies $0\in\cA$ as well as $\cA+L^0_+\subset\cA$.
\end{assumption}


Our assumptions are compatible with many relevant acceptability criteria.

\begin{example}
\label{ex: acceptance sets}
The following sets fulfill the defining properties of an acceptance set.
\begin{itemize}[leftmargin=*]
  \item {\em Expected shortfall}. Let $\alpha\in(0,1)$. For given $X\in L^0$ we define the Value at Risk of $X$ at level $\alpha$ as the negative of the upper $\alpha$-quantile of $X$, i.e.,
\[
\VaR_\alpha(X) := \inf\{x\in\R \,; \ \probp(X+x<0)\leq\alpha\} = -\inf\{x\in\R \,; \ \probp(X\leq x)>\alpha\}.
\]
The Expected Shortfall of $X$ at level $\alpha$ and the corresponding acceptance set are defined as
\[
\ES_\alpha(X) := \frac{1}{\alpha}\int_0^\alpha \VaR_p(X)dp, \ \ \ \ \ \ \cA_{\ES}(\alpha) := \{X\in L^0 \,; \ \ES_\alpha(X)\leq0\}.
\]
The set $\cA_{\ES}(\alpha)$ consists of those payoffs that are positive on average on the left tail beyond their upper $\alpha$-quantile. This acceptability criterion has been used in a pricing context by Cherny (2008).
  \item {\em Gain-loss ratios}. Let $\alpha\in\big(0,\frac{1}{2}\big]$. For a given $X\in L^0$ we define the expectile of $X$ at level $\alpha$ as the unique solution $e_\alpha(X)\in[-\infty,\infty]$ of the equation
\[
\alpha\E[(X-e_\alpha(X))^+]=(1-\alpha)\E[(e_\alpha(X)-X)^+]
\]
provided that either $X^+$ or $X^-$ is integrable, and $e_\alpha(X)=-\infty$ otherwise. The corresponding acceptance set is defined by
\[
\cA_e(\alpha) := \{X\in L^0 \,; \ e_\alpha(X)\geq0\} = \left\{X\in L^0 \, ; \ \frac{\E[X^+]}{\E[X^-]}\geq\frac{1-\alpha}{\alpha}\right\},
\]
with the convention $\frac{\infty}{\infty}=-\infty$ and $\frac{0}{0}=\infty$. This set consists of all the payoffs for which the ratio between the expected inflow of money (gains) and the expected outflow of money (losses) is sufficiently large. In particular, note that $\frac{1-\alpha}{\alpha}\geq1$, which implies that the expected gain must be at least as large as the the expected loss. This type of acceptability criterion has been investigated in a pricing context by Bernardo and Ledoit (2000), even though the link with expectiles was not discussed there.
  \item {\em Test scenarios}. Let $E\in\cF$ such that $\probp(E)>0$. The acceptance set given by
\[
\cA_E := \{X\in L^0 \,; \ X\one_E\geq0\}
\]
consists of all the payoffs that are positive on the event $E$. In this case, the elements of $E$ can be seen as pre-specified test or control scenarios and the acceptability criterion boils down to requiring a positive payment in each of these scenarios. Clearly, the set $\cA_E$ corresponds to the standard positive cone provided that we take $E=\Omega$ or more generally $\probp(E)=1$.
  \item {\em Test probabilities}. Let $\probq=(\probq_1,\dots,\probq_n)$ be a vector of probability measures on $(\Omega,\cF)$ that are absolutely continuous with respect to $\probp$. For a given vector $\alpha=(\alpha_1,\dots,\alpha_n)\in\R^n$ with nonpositive components we define the acceptance set
\[
\cA_\probq(\alpha) := \left\{X\in L^0 \,; \ \E\left[\frac{d\probq_i}{d\probp}X\right]\geq\alpha_i, \ \forall \ i\in\{1,\dots,n\}\right\},
\]
which consists of all the payoffs whose expected value under each of the pre-specified test probabilities is above the corresponding floor. The test probabilities may be designed, e.g., based on expert opinions or may correspond to appropriate distortions of the underlying probability measure $\probp$. This type of acceptability criterion has been investigated in a pricing context by Carr et al.\ (2001). In that paper, the probability measures used to define the acceptance set are called valuation test measures or stress test measures depending on whether the associated floor is zero or not.
  \item {\em Utility functions}. Let $u:\R\to[-\infty,\infty)$ be a nonconstant, increasing, concave function satisfying $u(0)=0$, which is interpreted as a von Neumann-Morgenstern utility function. For $\alpha\in(-\infty,0]$ we define the acceptance set by
\[
\cA_u(\alpha) := \{X\in L^0 \, ; \ \E[u(X)]\geq\alpha\},
\]
which consists of all the payoffs that yield a sufficiently large expected utility. In particular, the level $\alpha$ could coincide with some utility level, in which case $\cA_u(\alpha)$ would consist of all the payoffs that are preferable, from the perspective of the utility function $u$, to a pre-specified deterministic monetary loss. This type of acceptability criteria has been considered in a pricing context by \v{C}ern\'{y} and Hodges (2002), \v{C}ern\'{y} (2003), Kl\"{o}ppel and Schweizer (2007), and Arai (2011).
  \item {\em Stochastic dominance}. Recall that a random variable $X\in L^0$ with cumulative distribution function $F_X$ dominates a random variable $Y\in L^0$ with cumulative distribution function $F_Y$ in the sense of second-order stochastic dominance whenever for every $t\in\R$ we have
\[
\int_{-\infty}^t F_X(x)dx \leq \int_{-\infty}^t F_Y(y)dy.
\]
In this case, we write $X\succeq_{SSD}Y$. Now, fix $Z\in L^0$ with $0\succeq_{SSD}Z$ and define the acceptance set
\[
\cA_{SSD}(Z):=\{X\in L^0 \,; \ X\succeq_{SSD}Z\}.
\]
The reference payoff $Z$ may represent the terminal value of a pre-specified benchmark portfolio. Note that, by definition, we have $\E[Z]\leq0$. The use of stochastic dominance rules in pricing problems dates back at least to Levy (1985).
\end{itemize}
\end{example}


\subsection{Market-consistent prices with acceptable risk}

As already said, to determine a range of rational prices, the agent will identify among all attainable payoffs available in the market those that deliver an acceptable replication error and will use the corresponding market prices to assess whether a candidate buying price is too high or not. These prices are called market-consistent prices (with acceptable risk) and constitute the natural range of prices for a buyer who has access to the market, respects the existing portfolio constraints, and is willing to take up replication risk according to the chosen acceptance set. Indeed, if a price is not market consistent, then the agent can always invest that amount (or less) in the market to purchase an attainable payoff that ensures an acceptable replication error. This leads to the following formal definition, which was never explicitly formulated in the literature. We refer to Remark~\ref{rem: MCP} for a comparison with the literature.

\begin{definition}
A number $p\in\R$ is a {\em market-consistent price (with acceptable risk)} for $X\in L^0$ if:
\begin{enumerate}
\item[(1)] $p<\pi(Z)$ for every $Z\in\cM$ such that $Z-X\in\cA\setminus\{0\}$;
\item[(2)] $p\leq\pi(X)$ whenever $X\in\cM$.
\end{enumerate}
We denote by $\MCP(X)$ the set of market-consistent prices for $X$.
\end{definition}

The set of market-consistent prices for a payoff $X\in L^0$ is an interval that is bounded to the right. The upper bound is the natural generalization of the classical superreplication price to our setting, i.e.,
\[
\pi^+(X) := \inf\{\pi(Z) \,; \ Z\in\cM, \ Z-X\in\cA\}.
\]
We call $\pi^+(X)$ the {\em superreplication price (with acceptable risk)} of $X$. In other words, the superreplication price is the natural pricing threshold for a buyer who prices in a market-consistent way according to the underlying acceptance set. We record this observation in the next proposition.

\begin{proposition}
\label{prop: interval MCP}
For every $X\in L^0$ the set $\MCP(X)$ is an interval such that $\inf\MCP(X)=-\infty$ and $\sup\MCP(X)=\pi^+(X)$.
\end{proposition}
\begin{proof}
It is clear that $(-\infty,p)\subset\MCP(X)$ for every market-consistent price $p\in\MCP(X)$. Now, take any $p\in(-\infty,\pi^+(X))$ and note that, by definition of $\pi^+$, we have $p<\pi(Z)$ for every $Z\in\cM$ such that $Z-X\in\cA$. This shows that $p$ is a market-consistent price for $X$ and implies that $\pi^+(X)\leq\sup\MCP(X)$. Conversely, take an arbitrary market-consistent price $p\in\MCP(X)$. If $Z\in\cM$ is such that $Z-X\in\cA$, then $\pi(Z)\geq p$. Taking the infimum over such $Z$'s and the supremum over such $p$'s delivers the inequality $\pi^+(X)\geq\sup\MCP(X)$. This shows that $\pi^+(X)$ is the supremum of the set $\MCP(X)$.
\end{proof}

\begin{remark}
\label{rem: MCP}
(i) In line with our pricing problem, the notion of a market-consistent price is stated from a buyer's perspective. Following the same logic, one could define market-consistent prices from a seller's perspective and restrict the focus on prices that are simultaneously market consistent for both parties. One may wonder why we focus only on buyer's prices. From an economical perspective, this is because, as stressed above, the choice of the acceptance set is based on individual ``preferences'', implying that the general financial situation is that of a buyer and seller equipped with {\em different} acceptance sets. From a mathematical perspective, the buyer's and seller's problems are related to each other and one can easily adapt our results to obtain the corresponding results for seller's prices.

\smallskip

(ii) In the good deal pricing literature the focus is typically on superreplication prices and an explicit notion of a rational price is not explicitly discussed. The exception is Cherny (2008), where, in line with classical arbitrage pricing theory, rational prices are defined through extensions of the pricing rule preserving the absence of (suitably defined) good deals. Even though the pricing rule is not linear, the extension is assumed to be linear in the direction of the payoff that is ``added'' to the market. Our definition is not based on market extensions and does not require the absence of good deals, which, differently from the absence of arbitrage opportunities, is a debatable requirement; see Section~\ref{sect: acceptable deals}. Our approach extends that of Koch-Medina and Munari (2020) beyond the setting of frictionless markets and to general acceptance sets beyond the standard positive cone. We believe this approach is preferable from an economic perspective to the usual approach pursued in the good deal pricing literature.

\smallskip

(iii) Note that, in the definition of a market-consistent price, condition (1) need not imply condition (2), which is a natural requirement for a market-consistent price of an attainable payoff. The implication holds if, for instance, for every payoff $X\in\cM$ there exist a nonzero $U\in\cA$ and $c\in\R$ such that $X+\frac{1}{n}U\in\cM$ and $\pi(X+\frac{1}{n}U)\leq\pi(X)+\frac{1}{n}c$ for every $n\in\N$. In particular, this holds if $\cA$ and $\cM$ have nonzero intersection and $\pi$ and $\cM$ are both conic.
\end{remark}

\section{Good deals}
\label{sect: acceptable deals}

In this section we introduce the notion of good deals and discuss some important types of good deals. The absence of (a suitable class of) good deals will prove essential to establish our desired characterizations of market-consistent prices with acceptable risk. As explained below, this section deviates considerably from the path that is usually taken in the good deal pricing literature. In particular, we question the economic plausibility of the absence of good deals and replace it with the weaker condition of absence of scalable good deals.

\smallskip

A good deal is any nonzero acceptable payoff that is attainable and can be acquired at zero cost. As such, a good deal constitutes a natural generalization of an arbitrage opportunity, which corresponds to the situation where the acceptance set reduces to the standard positive cone. An important class of good deals is that of scalable good deals, i.e., payoffs that are good deals independently of their size. The notion of a good deal has appeared, sometimes with a slightly different meaning, under various names in the literature including good deal in Cochrane and Saa Requejo (2000), \v{C}ern\'{y} and Hodges (2002), Bj\"{o}rk and Slinko (2006),  Kl\"{o}ppel and Schweizer (2007), Bion-Nadal and Di Nunno (2013), Baes et al.\ (2020), good deal of first kind in Jaschke and K\"{u}chler (2001), good opportunity in Bernardo and Ledoit (2000), acceptable opportunity in Carr et al.\ (2001). The notion of a scalable good deal is a direct extension of that of a  scalable arbitrage opportunity introduced by Pennanen (2011a) and appeared, in a frictionless setting, in Baes et al.\ (2020). In this paper, we extend it to markets with frictions. The formal notions are recorded in the next definition, where we use recession cones and recession functionals as recalled in the appendix. Under our standing assumptions, $\cA^\infty$ and $\cM^\infty$ are the largest convex cones contained in $\cA$ and $\cM$, respectively. Similarly, $\pi^\infty$ is the smallest sublinear functional dominating $\pi$.

\begin{definition}
We say that a nonzero replicable payoff $X\in\cS$ is:
\begin{enumerate}
  \item[(1)] a {\em good deal (with respect to $\cA$)} if $X\in\cA\cap\cM$ and $\pi(X)\leq0$.
  \item[(2)] a {\em scalable good deal (with respect to $\cA$)} if $X\in\cA^\infty\cap\cM^\infty$ and $\pi^\infty(X)\leq0$.
  \item[(3)] a {\em strong scalable good deal (with respect to $\cA$)} if $X$ is a scalable good deal while $-X$ is not.
\end{enumerate}
We replace the term ``good deal'' with ``arbitrage opportunity'' whenever $\cA=L^0_+$.
\end{definition}


\begin{remark}
Note that if $X\in L^0$ is a strong scalable good deal, by definition there exists $\lambda>0$ such that $-\lambda X$ is not a good deal. However, this ``short'' position can be completely offset at zero cost by acquiring the attainable payoff $\lambda X$. This is what makes the scalable good deal ``strong''.
\end{remark}


It is clear that every strong scalable good deal is a scalable good deal, which in turn is a good deal. It is also clear that every (scalable) arbitrage opportunity is a (scalable) good deal. The absence of (scalable) good deals will be critical in our study of market-consistent prices. This condition plays the role of the absence of arbitrage opportunities in the classical arbitrage pricing theory. In that setting, an arbitrage opportunity constitutes an anomaly in the market because every rational agent will seek to exploit it thereby raising its demand until prices will also rise and the arbitrage opportunity will eventually vanish. The situation is quite different when we consider good deals as there might be no consensus across agents in the identification of a common criterion of acceptability, thereby casting doubts on the economic foundation of the absence of good deals. In our opinion, this crucial point has not been appropriately highlighted in the literature. The key observation about our paper is that the absence of good deals is not needed to develop our theory. Indeed, everything we have to ensure is that no (strong) scalable good deal exists. As shown by the next proposition, this weaker condition holds in a number of standard situations. In particular, we show that the absence of scalable arbitrage opportunities is sometimes sufficient to rule out scalable good deals. The condition $\cM^\infty\subset\cS_+$ is typically implied by caps on short positions (it holds if the payoffs of the basic securities are positive and the set of admissible portfolios is bounded from below so that short selling is possible but restricted for each security) while the condition $\cM^\infty=\{0\}$ is satisfied whenever there are caps on short and long positions alike (it is equivalent to the boundedness of the set of admissible portfolios); see Example~\ref{ex: market models}.

\begin{proposition}
\label{prop: no acc deals}
Assume that one of the following conditions holds:
\begin{enumerate}
  \item[(i)] $\cA^\infty=L^0_+$ and there exists no scalable arbitrage opportunity.
  \item[(ii)] $\cM^\infty\subset\cS_+$ and there exists no scalable arbitrage opportunity.
  \item[(iii)] $\cM^\infty=\{0\}$.
\end{enumerate}
Then, there exists no scalable good deal.
\end{proposition}
\begin{proof}
It is clear that no scalable good deal can exist if (iii) holds. Now, take a nonzero $X\in\cA^\infty\cap\cM^\infty$. Under either (i) or (ii), we have $X\in L^0_+$. Hence, we must have $\pi^\infty(X)>0$, for otherwise $X$ would be a scalable arbitrage opportunity. As a result, there cannot exist any scalable good deal.
\end{proof}


The next proposition records a simple equivalent condition for the absence of strong scalable good deals that will be used in the sequel. The condition is a one-period equivalent to the condition in Theorem 8 in Pennanen (2011b). In that paper the condition is expressed in terms of portfolios instead of payoffs and the acceptance set is the standard positive cone.

\begin{proposition}
\label{prop: characterization no strong acc}
There exists no strong scalable good deal if and only if $\cA^\infty\cap\{X\in\cM^\infty \,; \ \pi^\infty(X)\leq0\}$
is a vector space.
\end{proposition}
\begin{proof}
Let $\cN=\cA^\infty\cap\{X\in\cM^\infty \,; \ \pi^\infty(X)\leq0\}$. Note that $\cN$ is a convex cone. To prove the ``only if'' implication, assume that there exist no strong scalable good deal. This implies that $-\cN\subset\cN$, showing that $\cN$ is a vector space. To prove the ``if'' implication, assume that $\cN$ is a vector space. Take a nonzero $X\in\cA^\infty\cap\cM^\infty$ such that $\pi^\infty(X)\leq0$. Note that $X\in\cN$, so that $-X\in\cN$ as well. It follows that no strong scalable good deal can exist, concluding the proof.
\end{proof}


\section{Fundamental Theorem of Asset Pricing}
\label{sect: FTAP}

This section contains our main results. We establish a direct and a dual characterization of market-consistent pricing with acceptable risk. In line with classical results on arbitrage-free prices, the dual characterization relies on an appropriate extension of the Fundamental Theorem of Asset Pricing. Most of this section, including our dual results, is new and both extends and sharpens the corresponding results in the good deal pricing literature. We refer to the dedicated remarks for a detailed comparison with the literature.


\subsection{The reference payoff space}
\label{sect: payoff space}

The remainder of the paper is concerned with establishing characterizations of market-consistent prices with acceptable risk. In line with classical results on arbitrage-free prices, our characterizations will be obtained by means of topological methods. This will sometimes force us to restrict the set of payoffs we are able to price. This set is called {\em payoff space} and is denoted by $\cX\subset L^0$. Note that, for consistency, the payoff space should always contain those payoffs that already carry a market price, i.e., $\cS\subset\cX$. Of course, the natural choice is to take $\cX=L^0$ endowed with its canonical topology of convergence in probability. This choice, however, becomes problematic whenever we wish to apply duality theory. In this case, the payoff space has to be a strict subspace of $L^0$.

\smallskip

At first sight, the introduction of the payoff space at this stage may seem cumbersome and one may wonder why we have not focused our attention on $\cX$ instead of $L^0$ right from the beginning. In fact, this is the standard approach of the entire good deal pricing literature, where $\cX$ is taken to be a Lebesgue space or, more generally, a locally-convex topological vector space equipped with a partial order. Our approach has two main advantages. On the one side, we do not want to rule out the natural choice $\cX=L^0$. On the contrary, we aim to develop a pricing theory for general contracts as far as our methodology permits us. On the other side, the distinction between $L^0$ and $\cX$ makes it possible to unveil a number of critical aspects of good deal pricing that have never been highlighted in the literature.

\smallskip

A first aspect has to do with the assumption $\cS\subset\cX$. This inclusion implies that the choice of the payoff space has an impact on the range of basic securities, thereby inducing restrictions on the underlying market model. For example, the choice $\cX=L^1$ implies that each basic security has an integrable payoff, a condition that need not be satisfied by all realistic market models. Ensuring a flexible choice of the payoff space is therefore desirable from a modelling perspective. A second aspect has to do with the comparison between arbitrage pricing and good deal pricing. In arbitrage pricing one works under $\cX=L^0$, so that the explicit introduction of the payoff space is not necessary. The application of duality theory is nevertheless possible through a change of probability making the payoffs of the basic securities integrable. A similar trick cannot be reproduced in the setting of good deal pricing because the structure of the acceptance set, differently from the standard positive cone underpinning arbitrage pricing, is often disrupted through a change of probability. The analysis of the interplay between the payoff space and the acceptance set should thus be an integral part of good deal pricing. We refer to Remark~\ref{rem: on S contained in X} for more details about this critical, yet neglected, aspect of the theory.

\smallskip

We collect the standing assumptions on the payoff space below. As said, the choice $\cX=L^0$ endowed with the topology of convergence in probability is not ruled out for the moment. We will be careful to highlight when we are forced to restrict our analysis to a strict subspace of $L^0$.

\begin{assumption}
\label{assumption direct}
We assume that $\cX$ is a linear subspace of $L^0$ equipped with a linear Hausdorff topology. We also assume that $\cS\subset\cX$ and $\cA\cap\cX$ is closed.
\end{assumption}

\begin{remark}
\label{rem: on S contained in X}
(i) The assumption $\cS\subset\cX$ will be crucial to establish our characterizations of market-consistent prices. At the same time, we will sometimes need to apply duality theory and we will therefore have to assume that $\cX$ is a strict subspace of $L^0$. In this case, our assumption will force the payoffs of the basic traded securities to display some degree of regularity, e.g., integrability with respect to $\probp$. Inspired by standard arguments from arbitrage pricing theory, one may wonder whether this issue can be overcome by a simple change of probability. Indeed, define
\[
d\probq = \frac{c}{1+\sum_{i=1}^N|S_i|}d\probp, \ \ \ \ \ \  c=\E\left[\frac{1}{1+\sum_{i=1}^N|S_i|}\right],
\]
where $S_1,\dots,S_N\in L^0$ are the payoffs of the basic securities. It is immediate to see that the probability $\probq$ is equivalent to $\probp$ and every payoff in $\cS$ is integrable with respect to $\probq$. As a result, it would seem possible to work with $\cX=L^1(\probq)$, where $L^1(\probq)$ is the space of $\probq$-integrable random variables. This is precisely what is done in arbitrage pricing theory to make the application of duality theory possible; see, e.g., the proof of Theorem 1.7 in F\"{o}llmer and Schied (2016). The problem with this approach is that the acceptance set often depends explicitly on the natural probability $\probp$ and its (topological) properties are typically lost after we pass to $\probq$. Most importantly for our applications, the set $\cA\cap L^1(\probq)$ is seldom closed with respect to the norm topology of $L^1(\probq)$. Interestingly, this issue does not arise in arbitrage pricing theory because the acceptance set used there, namely $L^0_+$, is invariant with respect to changes of equivalent probability. More generally, the change of probability would not be problematic if the acceptance set is invariant with respect to changes of the numeraire. Unfortunately, as shown in Koch-Medina et al.\ (2017), numeraire invariance is only compatible with acceptance sets based on test scenarios as defined in Example~\ref{ex: acceptance sets}.

\smallskip

(ii) For technical reasons we need to require that the restriction of the acceptance set to $\cX$ is closed. This implies that the natural choice $\cX=L^0$ is feasible only if the chosen acceptance set is closed with respect to the topology of convergence in probability. This condition is sometimes satisfied, e.g., by $L^0_+$, but often fails. As a result, the choice of $\cX$ will generally depend on the underlying acceptance set.
\end{remark}


\subsection{A key auxiliary set}
\label{sect: superreplication C}

This subsection features a technical result that will play a crucial role in the sequel. We show that, under the absence of strong scalable good deals, the set
\[
\cC := \{(X,m)\in\cX\times\R \,; \ \exists Z\in\cM \,:\, Z-X\in\cA, \ \pi(Z)\leq-m\}
\]
is closed (in the natural product topology). This set consists of all the payoff-price couples $(X,m)\in\cX\times\R$ such that $X$ can be superreplicated with acceptable risk by means of admissible payoffs available in the market for less than $-m$. The set $\cC$ plays the same role that in classical arbitrage pricing theory is played by the set of payoffs that can be superreplicated at zero cost. To see the link, consider a frictionless market, i.e., a market where $\pi$ is linear and $\cM=\cS$, and assume that $\cA=L^0_+$. The set of payoffs that can be superreplicated at zero cost is given by
\[
\cK := \{X\in\cX \,; \ \exists Z\in\cS, \ Z\geq X, \ \pi(Z)\leq0\}.
\]
It is easily verified that, taking any $U\in\cS$ satisfying $\pi(U)=1$, we can rewrite $\cC$ as
\[
\cC = \{(X,m)\in\cX\times\R \,; \ X+mU\in\cK\}.
\]
In this classical setting, it is well known that the absence of arbitrage opportunities implies closedness of $\cK$ and, hence, of $\cC$. This is key to establish the classical Fundamental Theorem of Asset Pricing; see, e.g., F\"{o}llmer and Schied (2016). The closedness of $\cC$ in our general framework will allow us to establish a general version of the Fundamental Theorem in the next subsections.

\begin{lemma}
\label{lem: closedness C}
If there is no strong scalable good deal, then $\cC$ is closed and $(0,n)\notin\cC$ for some $n\in\N$.
\end{lemma}
\begin{proof}
Set $\cN=\{X\in\cA^\infty\cap\cM^\infty \,; \ \pi^\infty(X)\leq 0\}$ and denote by $\cN^\perp$ the orthogonal complement of $\cN$ in $\cS$. We claim that for every $(X,m)\in\cC$ there exists $Z\in\cM\cap\cN^\perp$ such that $Z-X\in\cA$ and $\pi(Z)\leq-m$. To see this, note that we find $W\in\cM$ such that $W-X\in\cA$ and $\pi(W)\leq-m$. We can write $W=W_{\cN}+W_{\cN^\perp}$ for unique elements $W_{\cN}\in\cN$ and $W_{\cN^\perp}\in\cN^\perp$. Note that $W_{\cN}$ belongs to $-\cN$ because the set $\cN$ is a vector space by Proposition~\ref{prop: characterization no strong acc}. Hence, setting $Z=W_{\cN^\perp}$, we infer that $Z=W-W_{\cN}\in\cM+\cM^\infty\subset\cM$ as well as $Z-X=(W-X)-W_{\cN}\in\cA+\cA^\infty\subset\cA$ by \eqref{eq: recession cones 1}. Moreover, $\pi(Z)=\pi(W-W_{\cN})\leq-m$ by combining \eqref{eq: recession cones 1} with \eqref{eq: recession cones 2}. This shows the desired claim.

\smallskip

Next, we establish closedness. To this end, take a net $(X_\alpha,m_\alpha)\subset\cC$ indexed on the directed set $(A,\succeq)$ and a point $(X,m)\in\cX\times\R$ and assume that $(X_\alpha,m_\alpha)\to(X,m)$. By assumption, we find a net $(Z_\alpha)\subset\cM$ such that $Z_\alpha-X_\alpha\in\cA$ and $\pi(Z_\alpha)\leq-m_\alpha$ for every $\alpha\in A$. Without loss of generality we can assume that $(Z_\alpha)\subset\cN^\perp$. Now, suppose that $(Z_\alpha)$ has no convergent subnet. In this case, we find a subnet of $(Z_\alpha)$ consisting of nonzero elements with strictly-positive diverging norms. (Indeed, it suffices to consider the index set $B=\{(\alpha,n) \,; \ \alpha\in A, \ n\in\N, \ \|Z_\alpha\|>n\}$ equipped with the direction defined by $(\alpha,n)\succeq(\beta,m)$ if and only if $\alpha\succeq\beta$ and $m\geq n$ and take $Z_{(\alpha,n)}=Z_\alpha$ for every $(\alpha,n)\in B$). We still denote this subnet by $(Z_\alpha)$. Since the unit sphere in $\cS$ is compact, we can assume that $\frac{Z_\alpha}{\|Z_\alpha\|} \to Z$ for a suitable nonzero $Z\in\cM^\infty$ by~\eqref{eq: recession cones 1}. As $(X_\alpha)$ is a convergent net by assumption,
\[
\frac{Z_\alpha-X_\alpha}{\|Z_\alpha\|} \to Z.
\]
This implies that $Z\in\cA^\infty$ again by~\eqref{eq: recession cones 1}. We claim that $\pi^\infty(Z)\leq0$. Otherwise, we must find $\lambda>0$ such that $\pi(\lambda Z)>0$. Without loss of generality we may assume that $\|Z_\alpha\|>\lambda$ for every $\alpha\in A$. Since $(m_\alpha)$ is a convergent net, we can use the lower semicontinuity and convexity of $\pi$ to get
\[
0 < \pi(\lambda Z) \leq \liminf_{\alpha}\pi\left(\frac{\lambda Z_\alpha}{\|Z_\alpha\|}\right) \leq \liminf_{\alpha}\frac{\lambda\pi(Z_\alpha)}{\|Z_\alpha\|} \leq \liminf_{\alpha}\frac{-\lambda m_\alpha}{\|Z_\alpha\|} = 0.
\]
This yields $\pi^\infty(Z)\leq0$. As a result, it follows that $Z$ belongs to $\cN$. However, this is not possible because $Z$ is a nonzero element in $\cN^\perp$. To avoid this contradiction, the net $(Z_\alpha)$ must admit a convergent subnet, which we still denote by $(Z_\alpha)$ for convenience. By closedness of $\cM$, the limit $Z$ also belongs to $\cM$. As we clearly have $Z_\alpha-X_\alpha\to Z-X$, it follows that $Z-X\in\cA$ by closedness of $\cA\cap\cX$. Moreover,
\[
\pi(Z) \leq \liminf_{\alpha}\pi(Z_\alpha) \leq \liminf_{\alpha}-m_\alpha = -m
\]
by lower semicontinuity of $\pi$. This shows that $(X,m)\in\cC$ and establishes that $\cC$ is closed.

\smallskip

Finally, we show that $(0,n)\notin\cC$ for some $n\in\N$. To this effect, assume to the contrary that for every $n\in\N$ there exists $Z_n\in\cA\cap\cM$ such that $\pi(Z_n)\leq-n$. If the sequence $(Z_n)$ is bounded, then we may assume without loss of generality that $Z_n\to Z$ for some $Z\in\cA\cap\cM$. The lower semicontinuity of $\pi$ implies $\pi(Z) \leq \liminf_{n\to\infty}\pi(Z_n) = -\infty$, which cannot hold. Hence, the sequence $(Z_n)$ must be unbounded. As argued above, we can assume that $(Z_n)\subset\cN^\perp$ without loss of generality. Moreover, we find a suitable subsequence, which we still denote by $(Z_n)$, that has strictly-positive divergent norms satisfying $\frac{Z_n}{\|Z_n\|}\to Z$ for some nonzero $Z$ belonging to $\cA^\infty\cap\cM^\infty$. We claim that $\pi^\infty(Z)\leq0$. Otherwise, we must find $\lambda>0$ such that $\pi(\lambda Z)>0$. Without loss of generality we may assume that $\|Z_n\|>\lambda$ for every $n\in\N$. The lower semicontinuity and convexity of $\pi$ imply
\[
0 < \pi(\lambda Z) \leq \liminf_{n\to\infty}\pi\left(\frac{\lambda Z_n}{\|Z_n\|}\right) \leq \liminf_{n\to\infty}\frac{\lambda\pi(Z_n)}{\|Z_n\|} \leq \liminf_{n\to\infty}\frac{-\lambda n}{\|Z_n\|} \leq 0.
\]
This shows that $\pi^\infty(Z)\leq0$ must hold. As a result, it follows that $Z$ belongs to $\cN$. However, this is not possible because $Z$ is a nonzero element in $\cN^\perp$. Hence, we must have $(0,n)\notin\cC$ for some $n\in\N$.
\end{proof}

\subsection{Direct characterization of market-consistent prices}

As observed in Proposition~\ref{prop: interval MCP}, the set of market-consistent prices with acceptable risk is an interval that is bounded from above by the corresponding superreplication price.
In this section we are concerned with establishing when the superreplication price is itself market consistent. This will yield a direct characterization of market-consistent prices. In Example \ref{ex: superreplication under replicability} we show that in general the superreplication price can be market consistent or not regardless of whether the underlying payoff is attainable or not. This is based on the following simple characterization of market consistency.

\begin{proposition}
\label{prop: characterization mcp superreplication}
For every $X\in\cX$ such that $\pi^+(X)\in\R$ we have $\pi^+(X)\in MCP(X)$ if and only if $(\cA+X)\cap\{Z\in\cM \,; \ \pi(Z)=\pi^+(X)\}\subset\{X\}$.
\end{proposition}
\begin{proof}
First, assume that $(\cA+X)\cap\{Z\in\cM \,; \ \pi(Z)=\pi^+(X)\}\subset\{X\}$. Then, for every $Z\in\cM$ satisfying $Z-X\in\cA\setminus\{0\}$ we must have $\pi(Z)>\pi^+(X)$. Since $\pi^+(X)\leq\pi(X)$ whenever $X\in\cM$, it follows that $\pi^+(X)\in\MCP(X)$, proving the ``if'' implication. Conversely, assume that $\pi^+(X)\in\MCP(X)$ and take any payoff $Z\in(\cA+X)\cap\cM$. If we happen to have $\pi(Z)=\pi^+(X)$, then $Z$ must be equal to $X$ by market consistency of $\pi^+(X)$. This proves the ``only if'' implication.
\end{proof}

The previous proposition shows that market consistency of the superreplication price is strongly linked with the attainability of the infimum in the definition of superreplication price. We therefore target sufficient conditions for attainability to hold. The next lemma suggests a strategy to tackle this problem.

\begin{lemma}
\label{lem: superreplication C}
For every $X\in\cX$ we have $\pi^+(X) = \inf\{m\in\R \,; \ (X,-m)\in\cC\}$.
\end{lemma}
\begin{proof}
For every $m\in\R$ we have $(X,-m)\in\cC$ if and only if there exists $Z\in\cM$ such that $Z-X\in\cA$ and $\pi(Z)\leq m$. As a result, we get
\[
\pi^+(X)
=
\inf\{m\in\R \,; \ \exists Z\in\cM \,:\, Z-X\in\cA, \ \pi(Z)\leq m\}
=
\inf\{m\in\R \,; \ (X,-m)\in\cC\}.\qedhere
\]
\end{proof}

The closedness criterion for the set $\cC$ established in Lemma~\ref{lem: closedness C} yields the following attainability result. 

\begin{proposition}
\label{prop: direct FTAP}
If there exists no strong scalable good deal, then for every $X\in\cX$ with $\pi^+(X)<\infty$ there exists $Z\in\cM$ such that $Z-X\in\cA$ and $\pi(Z)=\pi^+(X)$.
\end{proposition}
\begin{proof}
First of all, we note that $\pi^+$ is lower semicontinuous as, by virtue of Lemma~\ref{lem: closedness C}, $\cC$ is closed and the epigraph of $\pi^+$ coincides with $\{(X,m)\in\cX\times\R \, ; \ (X,-m)\in\cC\}$. Next, we claim that $\pi^+$ does not attain the value $-\infty$. To this end, note first that $\pi^+(0)>-\infty$ by Lemma \ref{lem: closedness C}. Since $\pi^+(0)\leq0$, it follows that $\pi^+$ is finite at $0$. It is readily seen that $\pi^+$ is convex. Hence, being lower semicontinuous, $\pi^+$ can never attain the value $-\infty$ on the space $\cX$.  
To show the desired attainability, take a payoff $X\in\cX$ such that $\pi^+(X)<\infty$. Since $\pi^+(X)$ is finite, it follows from the closedness of $\cC$ established in Lemma \ref{lem: closedness C} that the infimum in Lemma~\ref{lem: superreplication C} is attained. By definition of $\cC$, this implies that $\pi^+(X)=\pi(Z)$ for a suitable $Z\in\cM$ such that $Z-X\in\cA$.
\end{proof}


The next theorem provides a characterization of market-consistent prices under the assumption that the market does not admit strong scalable good deals. In this case, we show that for a payoff outside $\cM$ the superreplication price is never market consistent and, hence, the set of market-consistent prices is an open interval. For a payoff in $\cM$ the superreplication price may or may not be market consistent, so that the corresponding set of market-consistent prices may or may not be a closed interval.

\begin{proposition}[{\bf Direct characterization of market-consistent prices}]
\label{theo: characterization mcp superreplication}
If there exists no strong scalable good deal, then for every $X\in\cX$ we have $\MCP(X)\neq\emptyset$ and the following statements hold:
\begin{enumerate}
\item[(i)] If $X\in\cM$, then $\pi^+(X)\leq\pi(X)$ and both $\pi^+(X)\notin\MCP(X)$ and $\pi^+(X)\in\MCP(X)$ can hold.
\item[(ii)] If $X\in\cM$ and $\pi^+(X)\notin\MCP(X)$, then both $\pi^+(X)=\pi(X)$ and $\pi^+(X)<\pi(X)$ can hold.
\item[(iii)] If $X\in\cM$ and $\pi^+(X)\in\MCP(X)$, then $\pi^+(X)=\pi(X)$.
\item[(iv)] If $X\notin\cM$, then $\pi^+(X)\notin\MCP(X)$.
\end{enumerate}
The alternatives in (i) and (ii) can hold even if there exists no good deal.
\end{proposition}
\begin{proof}
It follows from Proposition~\ref{prop: direct FTAP} that for every $X\in\cX$ we must have $\pi^+(X)>-\infty$, showing that $\MCP(X)\neq\emptyset$. Now, take $X\in\cM$. Since $X-X=0\in\cA$, we easily infer from the definition of superreplication price that $\pi^+(X)\leq\pi(X)$. It is shown in Example~\ref{ex: superreplication under replicability} that all the situations in (i) and (ii) may hold (even if there exist no good deals). To establish (iii) and (iv), take an arbitrary $X\in\cX$ and assume that $\pi^+(X)\in\MCP(X)$. As Proposition~\ref{prop: direct FTAP} implies that $(\cA+X)\cap\{Z\in\cM \,; \ \pi(Z)=\pi^+(X)\}$ is not empty, it follows from Proposition~\ref{prop: characterization mcp superreplication} that $X$ must belong to $\cM$ and that the infimum in the definition of superreplication price must be attained by $X$ alone, establishing the desired implications.
\end{proof}

\begin{example}
\label{ex: superreplication under replicability}
Let $\Omega=\{\omega_1,\omega_2\}$ and assume that $\cF$ is the power set of $\Omega$ and that $\probp$ is specified by $\probp(\omega_1)=\probp(\omega_2)=\frac{1}{2}$. In this simple setting, we take $\cX=L^0$ and identify every element of $\cX$ with a vector of $\R^2$. Set $\cS=\R^2$ and consider the acceptance set defined by
\[
\cA = \{(x,y)\in\R^2 \,; \ y\geq\max\{-x,0\}\}.
\]

\smallskip

(i) Set $\pi(x,y)=\max\{2x+y,x+2y\}$ for every $(x,y)\in\R^2$ and $\cM=\R^2$. It is immediate to verify that no good deal exists. Set $X=(-2,1)\in\cM$ and observe that $\pi^+(X)=0$ and
\[
(\cA+X)\cap\{Z\in\cM \,; \ \pi(Z)=0\}=\{X\}.
\]
It follows from Proposition \ref{prop: characterization mcp superreplication} that $\pi^+(X)\in MCP(X)$. Next, take $Y=(1,-2)\in\cM$. In this case, an explicit calculation shows that
\[
\pi^+(Y) = \inf_{x\in\R}\max\{2x-2+\max\{1-x,0\},x-4+2\max\{1-x,0\}\} = -\frac{3}{2}.
\]
Moreover, setting $W=(-\frac{1}{2},-\frac{1}{2})\in\cM$, we have
\[
(\cA+Y)\cap\bigg\{Z\in\cM \,; \ \pi(Z)=-\frac{3}{2}\bigg\}=\{W\}.
\]
It follows from Proposition \ref{prop: characterization mcp superreplication} that $\pi^+(Y)\notin\MCP(Y)$. Note also that $\pi(X)=\pi^+(X)$ and $\pi(Y)>\pi^+(Y)$.

\smallskip

(ii) Set $\pi(x,y)=\max\{x+y,x+2y\}$ for every $(x,y)\in\R^2$ and $\cM=\{(x,y)\in\R^2\,; \ x\leq1\}$. Observe that no good deal exists. Set $X=(1,-1)\in\cM$ and $Y=(2,-2)\notin \cM$. Then, $\pi^+(X)=\pi^+(Y)=0$ and
\[
(\cA+X)\cap\{Z\in\cM \,; \ \pi(Z)=0\}=(\cA+Y)\cap\{Z\in\cM \,; \ \pi(Z)=0\}=\{\lambda X \,; \ \lambda\in[0,1]\}.
\]
It follows from Proposition \ref{prop: characterization mcp superreplication} that $\pi^+(X)\notin \MCP(X)$ and $\pi^+(Y)\notin \MCP(Y)$. Note also that $\pi(X)=0$ so that $\pi(X)=\pi^+(X)$.

\smallskip

(iii) Set $\pi(x,y)=e^x-1$ for every $(x,y)\in\R^2$ and $\cM=\R\times\R_+$. Any $X\in\cX$ satisfies $\pi^+(X)=-1$ and
\[
(\cA+X)\cap\{Z\in\cM \,; \ \pi(Z)=-1\} = \emptyset.
\]
It follows from Proposition \ref{prop: characterization mcp superreplication} that $\pi^+(X)\in MCP(X)$ regardless of whether $X$ belongs to $\cM$ or not. Note that in this case there exist strong scalable good deals.
\end{example}

The previous proposition unveils a stark contrast between our general setting and the classical frictionless setting. In a frictionless market, the superreplication price of every replicable payoff is market consistent and coincides with the associated replication cost. In our case, for an attainable payoff, the superreplication price may be strictly lower than the associated replication cost. This is in line with the findings in Bensaid et al.\ (1992), where the focus was on a multi-period Cox-Ross-Rubinstein model with proportional transaction costs and no portfolio constraints and the acceptance set was taken to be the standard positive cone. As explained in that paper, the discrepancy between the superreplication price and the replication cost is a direct consequence of the fact that trading is costly and it may therefore ``pay to weigh the benefits of replication against those of potential savings on transaction costs''. What also follows from the previous result and was only implicitly highlighted in Bensaid et al.\ (1992) is that, contrary to the frictionless case, the superreplication price of an attainable payoff and, {\em a fortiori}, its replication cost may fail to be market consistent. This is another implication of transaction costs, which allow the infimum in the definition of superreplication price to be attained by multiple replicable payoffs even if the market admits no good deals. Motivated by this discussion, we provide sufficient conditions for the replication cost of a payoff in $\cM$ to be market consistent and, hence, to coincide with the corresponding superreplication price. More precisely, we show that this holds for every payoff with ``zero bid-ask spread'' provided the market admits no good deals.

\begin{proposition}
If there exists no good deal, then $\pi(X)=\pi^+(X)\in\MCP(X)$ for every $X\in\cM\cap(-\cM)$ such that $\pi(-X)=-\pi(X)$.
\end{proposition}
\begin{proof}
Take an arbitrary $X\in\cM\cap(-\cM)$ such that $\pi(-X)=-\pi(X)$. Since $\pi^+(X)$ is the supremum of the set $\MCP(X)$ and $\pi^+(X)\leq\pi(X)$, it suffices to show that $\pi(X)\in\MCP(X)$. To this effect, take any $Z\in\cM$ satisfying $Z-X\in\cA\setminus\{0\}$. Note that $\frac{1}{2}Z-\frac{1}{2}X=\frac{1}{2}(Z-X)+\frac{1}{2}0\in\cA\cap\cM$. As a result, the absence of good deals implies that
\[
0 < \pi\left(\frac{1}{2}Z-\frac{1}{2}X\right) \leq \frac{1}{2}\pi(Z)+\frac{1}{2}\pi(-X) = \frac{1}{2}\pi(Z)-\frac{1}{2}\pi(X).
\]
This yields $\pi(X)<\pi(Z)$ and proves that $\pi(X)$ is a market-consistent price for $X$.
\end{proof}

\subsection{Consistent price deflators}

In this subsection we start our journey towards a dual characterization of market-consistent prices with acceptable risk. As already mentioned, a key step is to establish the appropriate extension of the Fundamental Theorem of Asset Pricing. Both results will be expressed in terms of suitable dual elements, called consistent price deflators. Here, consistency refers to the acceptance set. We mainly distinguish between two types of price deflators, namely consistent and strictly consistent ones. These notions are encountered in the literature under special assumptions on the market model and/or on the acceptance set. In a frictionless setting, a consistent price deflator corresponds to a representative state pricing function in Carr et al.\ (2001) and to a Riesz density of a no-good-deal pricing functional in \v{C}ern\'{y} and Hodges (2002). In a market with proportional frictions, it corresponds to a Riesz density of an underlying frictionless pricing rule in Jouini and Kallal (1995), to a consistent price system in Jaschke and K\"{u}chler (2001), to a consistent pricing kernel in Staum (2004), and is related to a risk-neutral measure in Cherny (2008). In a market with nonproportional frictions, it corresponds to a marginal price deflator in Pennanen (2011a). Strictly consistent price deflators have been considered in Jouini and Kallal (1995), \v{C}ern\'{y} and Hodges (2002), and Pennanen (2011a). Note that the acceptance set in Jouini and Kallal (1995) and Pennanen (2011a) is the standard positive cone. The formal definition of a price deflator is as follows.

\begin{definition}
\label{def: pricing density}
A random variable $D\in L^0$ is a {\em price deflator} if the following conditions hold:
\begin{enumerate}
  \item[(1)] $DX\in L^1$ for every $X\in\cS$.
  \item[(2)] $\sup\{\E[DX]-\pi(X) \,; \ X\in\cM\}<\infty$.
\end{enumerate}
In this case, we say that $D$ is:
\begin{enumerate}
  \item[(3)] {\em weakly consistent} if $\inf\{\E[DX] \,; \ X\in\cA\cap\cX\}>-\infty$.
  \item[(4)] {\em consistent} if $\E[DX]\geq0$ for every $X\in\cA\cap\cX$.
  \item[(5)] {\em strictly consistent} if $\E[DX]>0$ for every nonzero $X\in\cA\cap\cX$.
\end{enumerate}
\end{definition}


It should be clear that a price deflator is a natural extension of a classical price deflator to our market with frictions. To illustrate this, consider a price deflator $D\in L^0$ and define $\cL=\{X\in L^0 \,; \ DX\in L^1\}$. Note that every replicable payoff belongs to the vector space $\cL$. Moreover, define $\psi(X)=\E[DX]$ for $X\in\cL$. By definition, there exists a constant $\gamma_{\pi,\cM}\geq0$ such that for every attainable payoff $X\in\cM\cap(-\cM)$
\[
-\pi(-X)-\gamma_{\pi,\cM}\leq\psi(X)\leq\pi(X)+\gamma_{\pi,\cM}.
\]
The functional $\psi$ can therefore be viewed as the pricing rule of an ``artificial'' frictionless market where every payoff in $\cL$ is ``replicable'' and the attainable payoffs are ``priced'', up to a suitable enlargement, consistently with their market bid-ask spread. No enlargement is needed when $\psi$ is already dominated from above by $\pi$. This happens, for instance, if both $\pi$ and $\cM$ are conic in the first place. In particular, this holds if $\pi$ is linear and $\cM$ coincides with the entire $\cS$, in which case $\psi$ is a linear extension of the pricing rule beyond the space of replicable payoffs. This shows that, in a frictionless setting, the notion of a price deflator boils down to the classical notion from arbitrage pricing theory. Consistency with the acceptance set is, of course, specific to good deal pricing theory. The interpretation is as follows. If $D$ is weakly consistent, then we find a constant $\gamma_\cA\leq0$ such that for every acceptable payoff $X\in\cA\cap\cX\cap\cL$
\[
\psi(X) \geq \gamma_\cA.
\]
This means that prices of acceptable payoffs in the ``artificial'' frictionless market with pricing rule $\psi$ cannot be arbitrarily negative. A simple situation where such ``artificial'' prices are nonnegative is when $\cA$ is a cone in the first place. In this case, weak consistency is equivalent to consistency. In particular, if $\cA$ is taken to be the standard positive cone, then (strict) consistency boils down to the (strict) positivity of $\psi$. Hence, consistency with the acceptance set requires that the pricing rule in the ``artificial'' frictionless market assigns prices to acceptable payoffs that are bounded from below, positive, or strictly positive depending on the type of consistency. If the acceptance set is the standard positive cone, then (strict) consistency boils down to (strict) positivity. This shows that a (strictly) consistent price deflator is a direct extension of a (strictly) positive price deflator, or equivalently of an (equivalent) martingale measure, in the classical theory. We summarize this discussion in the following proposition, which highlights the role of conicity in simplifying the formulation of a consistent price deflator.

\begin{proposition}
\label{prop: pricing density}
Let $D\in L^0$ be a price deflator. Then, the following statements hold:
\begin{enumerate}
  \item[(i)] $\E[DX]\leq\pi(X)$ for every $X\in\cM^\infty$ such that $\pi$ is conic on $\cone(X)$.
  \item[(ii)] $\E[DX]=\pi(X)$ for every $X\in\cM^\infty\cap(-\cM^\infty)$ such that $\pi$ is linear on $\Span(X)$.
\end{enumerate}
If $D$ is weakly consistent, then the following statement holds:
\begin{enumerate}
  \item[(iii)] $\E[DX]\geq0$ for every $X\in\cA^\infty\cap\cX$.
\end{enumerate}
\end{proposition}
\begin{proof}
Take an arbitrary $X\in\cX$. Since $\Span(X)=\cone(X)\cup\cone(-X)$, it is clear that (i) implies (ii). To prove (i), assume that $X\in\cM^\infty$ and $\pi$ is conic on $\cone(X)$. Then, by definition of a pricing density,
\[
\sup_{n\in\N}\{n(\E[DX]-\pi(X))\} = \sup_{n\in\N}\{\E[D(nX)]-\pi(nX)\} < \infty.
\]
This is only possible if $\E[DX]-\pi(X)\leq0$, showing the desired claim. Finally, to establish (iii), assume that $D$ is weakly consistent and $X\in\cA^\infty$. Then, by definition of weak consistency,
\[
\inf_{n\in\N}\{n\E[DX]\} = \inf_{n\in\N}\E[D(nX)] > -\infty.
\]
This is only possible if $\E[DX]\geq0$, proving the desired claim and concluding the proof.
\end{proof}

Later, we will show that, contrary to the focus on consistent price deflators of the bulk of the literature, strictly consistent price deflators are the right dual objects to use in order to obtain a version of the Fundamental Theorem of Asset Pricing in a good deal pricing setting. For the time being, we show that the existence of strictly consistent price deflators always implies the absence of scalable good deals. However, contrary to the classical frictionless setting, it does not generally imply the absence of good deals unless the price deflators satisfy suitable extra assumptions.

\begin{proposition}
\label{prop: density implies no good deal}
If there exists a strictly consistent price deflator $D\in L^0$, then there exists no scalable good deal. If, additionally, $\E[DX]\leq\pi(X)$ for every $X\in\cM$, then there exists no good deal either.
\end{proposition}
\begin{proof}
Take a nonzero payoff $X\in\cA\cap\cM^\infty$. To show that no scalable good deal exists, we have to show that $\pi^\infty(X)>0$. To this effect, note that, by definition of a price deflator,
\[
\sup_{n\in\N}\{n(\E[DX]-\pi^\infty(X))\} =
\sup_{n\in\N}\{\E[D(nX)]-\pi^\infty(nX)\} \leq \sup_{n\in\N}\{\E[D(nX)]-\pi(nX)\} < \infty,
\]
where we used that $\pi^\infty$ dominates $\pi$. This is only possible if $\E[DX]-\pi^\infty(X)\leq0$. As a result, we obtain $\pi^\infty(X)\geq\E[DX]>0$. Next, assume that $\E[DX]\leq\pi(X)$ for every payoff $X\in\cM$ and take a nonzero payoff $X\in\cA\cap\cM$. Then, $\pi(X)\geq\E[DX]>0$, showing that no good deal exists.
\end{proof}

\begin{example}
\label{ex: scpricingdens with acc deals}
Let $\Omega=\{\omega_1,\omega_2\}$ and assume that $\cF$ is the power set of $\Omega$ and that $\probp$ is specified by $\probp(\omega_1)=\probp(\omega_2)=\frac{1}{2}$. In this simple setting, we take $\cX=L^0$ and identify every element of $\cX$ with a vector of $\R^2$. Set $\cS=\R^2$ and consider the acceptance set defined by
\[
\cA = \{(x,y)\in\R^2 \,; \ y\geq\max\{-x,0\}\}.
\]
We show that the existence of strictly consistent price deflators is not sufficient to rule out good deals. In view of the previous proposition, this can occur only if either the pricing rule or the set of attainable payoffs fails to be conic and the supremum in Definition~\ref{def: pricing density} is strictly positive. We provide an example in both cases.

\smallskip

(i) Set $\pi(x,y)=x+y^2$ for every $(x,y)\in\R^2$ and $\cM=\R^2$. Note that $\cM$ is conic while $\pi$ is not. It is clear that $D=(2,4)$ is a strictly consistent price deflator. In particular, we have
\[
\sup_{X\in\cM}\{\E[DX]-\pi(X)\} = \sup_{y\in\R}\{2y-y^2\} = 1.
\]
However, $X=(-1,1)\in\cA\cap\cM$ satisfies $\pi(X)=0$ and is thus a good deal.

\smallskip

(ii) Set $\pi(x,y)=x+y$ for every $(x,y)\in\R^2$ and $\cM=\{(x,y)\in\R^2 \,; \ x\geq-1, \ 0\leq y\leq 1\}$. Note that $\pi$ is conic while $\cM$ is not. It is clear that $D=(2,4)$ is a strictly consistent price deflator. In particular,
\[
\sup_{X\in\cM}\{\E[DX]-\pi(X)\} = \sup_{0\leq y\leq1}y = 1.
\]
However, $X=(-1,1)\in\cA\cap\cM$ satisfies $\pi(X)=0$ and is thus a good deal.
\end{example}


\subsection{The reference set of price deflators}

We turn to the more challenging problem of investigating if and under which assumptions the converse implication holds, i.e., the absence of (scalable) good deals implies the existence of strictly consistent price deflators. To this effect, we rely on duality theory and we therefore have to choose a suitable topology on the reference payoff space. As remarked below, our framework is flexible enough to accommodate the standard model spaces. We refer to the appendix for the necessary details on weak topologies. 

\begin{assumption}
\label{standing assumption}
We denote by $\cX'$ a linear subspace of $L^0$. We assume that $\cX$ and $\cX'$ contain $L^\infty$ and satisfy $XY\in L^1$ for all $X\in\cX$ and $Y\in\cX'$. These spaces are in separating duality through the bilinear form $(X,Y) \mapsto \E[XY]$. The topology on $\cX$ fixed in Assumption~\ref{assumption direct} is taken to be $\sigma(\cX,\cX')$. Similarly, we equip $\cX'$ with the topology $\sigma(\cX',\cX)$. In addition, we assume that $\cX'$ is the norm dual of a normed space $\cY\subset L^0$ (which need not coincide with $\cX$) and that $\sigma(\cX',\cX)$ is weaker than the associated weak-star topology $\sigma(\cX',\cY)$.
\end{assumption}

\begin{remark}
\label{rem: setting}
(i) Under our assumption the payoff space $\cX$ could be any Lebesgue space or, more generally, any Orlicz space and the dual space $\cX'$ could be $L^\infty$, in which case $\cY$ is taken to be $L^1$.

\smallskip

(ii) Note that, under our assumption, both topologies $\sigma(\cX,\cX')$ and $\sigma(\cX',\cX)$ are Hausdorff and locally convex. Hence, the standard machinery of convex duality applies to them.

\smallskip

(iii) Under our standing assumptions the set $\cA\cap\cX$ has to be $\sigma(\cX,\cX')$-closed. For the common payoff spaces and acceptance sets, this is fulfilled even in the (generally restrictive) situation where $\cX'$ is a small space. For concreteness, let $(\Omega,\cF,\probp)$ be nonatomic and let $\cX$ be an Orlicz space. Moreover, let $\cX'=L^\infty$. The set $\cA\cap\cX$ is closed with respect to $\sigma(\cX,\cX')$ in any of the following cases:
\begin{enumerate}
  \item[(a)] $\cA\cap L^1$ is closed with respect to the norm topology of $L^1$.
  \item[(b)] $\cA$ is either law invariant under $\probp$ or surplus invariant, and for all $(X_n)\subset\cA\cap\cX$ and $X\in\cX$ such that $X_n\to X$ $\probp$-almost surely and $\sup_{n\in\N}|X_n|\in\cX$ it follows that $X\in \cA$.
\end{enumerate}
The condition in point (a) clearly implies $\sigma(\cX,\cX')$-closedness of $\cA\cap\cX$. In point (b), law invariance stipulates that acceptability is only driven by the probability distribution of a payoff while surplus invariance, introduced in Koch-Medina et al.\ (2015) and studied more thoroughly in Koch-Medina et al.\ (2017), stipulates that acceptability is only driven by the downside profile of a payoff. The closedness under dominated $\probp$-almost sure convergence is sometimes referred to as Fatou closedness. In these cases the desired $\sigma(\cX,\cX')$-closedness of $\cA\cap\cX$ follows from the results in Svindland (2010) and Gao et al.\ (2018) under law invariance and from those in Gao and Munari (2020) under surplus invariance.
\end{remark}


We define the sets of weakly and strictly consistent price deflators belonging to $\cX'$ as follows:
\[
\cD := \{D\in\cX' \,; \ \mbox{$D$ is a weakly consistent price deflator}\},
\]
\[
\cD_{str} := \{D\in\cX' \,; \ \mbox{$D$ is a strictly consistent price deflator}\}.
\]
It is also convenient to introduce the maps $\gamma_{\pi,\cM}:\cX'\to(-\infty,\infty]$ and $\gamma_\cA:\cX'\to[-\infty,\infty)$ defined by
\[
\gamma_{\pi,\cM}(Y) := \sup_{X\in\cM}\{\E[XY]-\pi(X)\},
\]
\[
\gamma_\cA(Y) := \inf_{X\in\cA\cap\cX}\E[XY].
\]
Note that $\gamma_{\pi,\cM}$ coincides with the conjugate function of the restriction to $\cM$ of the pricing rule $\pi$ whereas $\gamma_\cA$ is, up to a sign, the support function of the set $-(\cA\cap\cX)$. These maps appear in the definition of a weakly consistent price deflator. A key role in our analysis is again played by the set $\cC$ introduced in Section~\ref{sect: superreplication C}. In particular, weakly consistent price deflators appear naturally in the dual representation of $\cC$. We denote by $\cl(\cC)$ the closure of $\cC$ (with respect to the natural product topology on $\cX\times\R$) and refer to the appendix for the notation on support functions and barrier cones.

\begin{lemma}
\label{lem: elementary properties C}
The sets $\cC$ and $\cD$ are convex and the following statements hold:
\begin{enumerate}
  \item[(i)] $-((\cA\cap\cX)\times\R_+)\subset\cC$ and $\barr(\cC)\subset\cX'_+\times\R_+$.
  \item[(ii)] $\sigma_\cC(Y,1)=\gamma_{\pi,\cM}(Y)-\gamma_\cA(Y)$ for every $Y\in\cX'$.
  \item[(iii)] $\cD=\{Y\in\cX'_+ \,; \ \sigma_\cC(Y,1)<\infty\} = \{Y\in\cX'_+ \,; \ (Y,1)\in\barr(\cC)\}$.
  \item[(iv)] If $(0,n)\notin\cl(\cC)$ for some $n\in\N$, then we can represent $\cl(\cC)$ as
\[
\cl(\cC) = \bigcap_{Y\in\cD}\{(X,m)\in\cX\times\R \,; \ \E[XY]+m\leq\gamma_{\pi,\cM}(Y)-\gamma_\cA(Y)\}.
\]
\end{enumerate}
\end{lemma}
\begin{proof}
The convexity of $\cC$ and $\cD$ is clear. Points (i), (ii), and (iii) follow easily from rewriting $\cC$ as
\begin{equation*}
    \label{eq: C difference}
    \cC = \{(Z,m)\in\cM\times\R \, ; \ \pi(Z)\leq-m\}-(\cA\cap\cX)\times\R_+.
\end{equation*}
Note that no problems with nonfinite values arise as $0\in\cM$, $\pi(0)=0$, and $\cA$ contains the cone of positive random variables.
To show (iv), assume that $\cl(\cC)$ is strictly contained in $\cX\times\R$. The dual representation of closed convex sets recorded in Theorem 7.51 of Aliprantis and Border (2006) yields
\begin{equation}
\label{eq: repr C}
\cl(\cC) = \bigcap_{(Y,r)\in\cX'\times\R}\{(X,m)\in\cX\times\R \,; \ \E[XY]+mr\leq \sigma_\cC(Y,r)\}.
\end{equation}
Here, we have used that $\sigma_{\cl(\cC)}=\sigma_\cC$. We claim that $\barr(\cC)\cap(\cX'\times(0,\infty))\neq\emptyset$. To show this, take $n\in\N$ such that $(0,n) \notin \cl(\cC)$. Then, it follows from~\eqref{eq: repr C} that there must exist $(Y,r)\in\barr(\cC)$ satisfying $nr = \E[0\cdot Y]+nr > \sigma_\cC(Y,r) \geq 0$. This establishes the desired claim. Now, recall from point (i) that $\barr(\cC)\subset\cX'_+\times\R_+$. Since $\sigma_\cC$ is sublinear and $\barr(\cC)$ is a convex cone, it follows that
\[
\cl(\cC) = \bigcap_{Y\in\cX'_+}\{(X,m)\in\cX\times\R \,; \ \E[XY]+m\leq \sigma_\cC(Y,1)\}.
\]
The desired representation is now a direct consequence of point (ii).
\end{proof}


\subsection{Fundamental theorem of asset pricing}

The key tool to establish the Fundamental Theorem of Asset Pricing is the following convenient version of the classical results by Yan (1980) and Kreps (1981). We refer to Clark (1993), Jouini et al.\ (2005), Rokhlin (2005), Cassese (2007), Rokhlin (2009), and Gao and Xanthos (2017) for a variety of versions of the same principle. A simple inspection of our formulation shows that, taking $\cL'=\cD$, the theorem delivers existence of price deflators that assign a strictly positive ``price'' to every nonzero payoff in $\cL$. Depending on the choice of $\cL$ we obtain different types of price deflators. The choice $\cL=\cX_+$ yields strictly positive price deflators whereas the choice $\cL=\cA\cap\cX$ yields strictly consistent ones.

\begin{theorem}[{\bf Kreps-Yan}]
\label{theo: kreps yan in our setting}
Let $\cL\subset\cX$ and $\cL'\subset\cX'$ and assume that the following properties hold:
\begin{enumerate}
  \item[(i)] Completeness: For every sequence $(Y_n)\subset\cL'$ there exist a sequence $(\lambda_n)\subset(0,\infty)$ and $Y\in\cL'$ such that $\sum_{k=1}^n\lambda_kY_k\to Y$.
  \item[(ii)] Countable separation: There exists a sequence $(Y_n)\subset\cL'\cap(-\barr(\cone(\cL)))$ such that for every nonzero $X\in\cL$ we have $\E[XY_n]>0$ for some $n\in\N$.
\end{enumerate}
Then, there exists $Y\in\cL'$ such that $\E[XY]>0$ for every nonzero $X\in\cL$.
\end{theorem}
\begin{proof}
By the countable separation property, there exists a sequence $(Y_n)\subset\cL'\cap(-\barr(\cone(\cL)))$ such that for every nonzero $X\in\cL$ we have $\E[XY_n]>0$ for some $n\in\N$. In particular, note that $\E[XY_n]\geq0$ for all $X\in\cL$ and $n\in\N$ because $(Y_n)\subset-\barr(\cone(\cL))$. Moreover, by the completeness property, there exist a sequence $(\lambda_n)\subset(0,\infty)$ and $Y\in\cL'$ such that $\sum_{k=1}^n\lambda_kY_k\to Y$. It is immediate to see that $\E[XY]>0$ for every nonzero $X\in\cL$.
\end{proof}


\begin{remark}
The preceding theorem holds for every pair of vector spaces $\cX$ and $\cX'$ equipped with a bilinear mapping $\langle\cdot,\cdot\rangle:\cX\times\cX'\to\R$. In this respect, our statement is a minor extension of the abstract version of the result obtained by Jouini et al.\ (2005). In that paper, the set $\cL$ was assumed to be a pointed convex cone satisfying $\cL-\cL=\cX$ and the dual set $\cL'$ was taken to coincide with $-\barr(\cL) = \{Y\in\cX' \,; \ \langle X,Y\rangle\geq0, \ \forall X\in\cL\}$. Incidentally, note that pointedness is automatically implied by the countable separation property (regardless of the special choice of $\cL$). 
\end{remark}

The ``conification'' in the Kreps-Yan theorem leads us to work with the modified acceptance set
\[
\cK(\cA) := \cl(\cone(\cA\cap\cX))+L^0_+
\]
where $\cl$ is the closure operator with respect to the reference topology $\sigma(\cX,\cX')$. A similar conification was considered in \v{C}ern\'{y} and Hodges (2002) and Staum (2004) and is necessary to obtain a version of the Fundamental Theorem that applies to nonconic acceptance sets. The next proposition contains a list of useful properties of this enlarged acceptance set.

\begin{proposition}
\label{prop: conified A}
The set $\cK(\cA)$ satisfies the properties in Assumption~\ref{ass: acceptance set} provided it is a strict subset of $L^0$. Moreover, it is a cone and satisfies $\cK(\cA)\cap\cX=\cl(\cone(\cA)\cap\cX)$. In particular, if $\cA$ is a cone, then $\cK(\cA)\cap\cX=\cA\cap\cX$.
\end{proposition}
\begin{proof}
It is readily seen that $\cK(\cA)$ satisfies the properties in Assumption~\ref{ass: acceptance set} and is a cone. Note that $\cK(\cA)\cap\cX=\cl(\cone(\cA\cap\cX))+\cX_+$. Hence, it remains to show that $\cl(\cone(\cA\cap\cX))+\cX_+\subset\cl(\cone(\cA\cap\cX))$. To this end, take arbitrary $X\in\cl(\cone(\cA\cap\cX))$ and $U\in\cX_+$. By assumption, we find nets $(X_\alpha)\subset\cA\cap\cX$ and $(\lambda_\alpha)\subset\R_+$ such that $\lambda_\alpha X_\alpha\to X$. Clearly, $\lambda_\alpha X_\alpha+U\to X+U$. We conclude by showing that for every $\alpha$ we have $\lambda_\alpha X_\alpha+U\in\cone(\cA)$. This is obvious if $\lambda_\alpha=0$ because $U\in\cX_+\subset\cA$. Otherwise, assume that $\lambda_\alpha>0$. In this case, we have $X_\alpha+\frac{1}{\lambda_\alpha}U\in\cA+\cX_+\subset\cA$ by monotonicity of $\cA$. Hence, it follows that $\lambda_\alpha X_\alpha+U=\lambda_\alpha(X_\alpha+\frac{1}{\lambda_\alpha}U)\in\cone(\cA)$. This concludes the proof.
\end{proof}

To establish the existence of strictly consistent price deflators through the Kreps-Yan theorem we have to verify the completeness and countable separation properties. We start by showing that completeness always holds in our setting. This is a direct consequence of the fact that, by assumption, the space $\cX'$ is a norm dual and $\sigma(\cX',\cX)$ is weaker than the corresponding weak-star topology.

\begin{proposition}
\label{prop: completeness property}
For every sequence $(Y_n)\subset\cD$ there exist a sequence $(\lambda_n)\subset(0,\infty)$ and $Y\in\cD$ such that $\sum_{k=1}^n\lambda_kY_k\to Y$.
\end{proposition}
\begin{proof}
Recall that $\cD\subset\cX'_+$ by Lemma~\ref{lem: elementary properties C} and note that $\sigma_\cC(Y,1)\geq0$ for every $Y\in\cD$. Moreover, recall that $\cX'$ is a norm dual and denote by $\|\cdot\|_{\cX'}$ the corresponding dual norm. Let $S_n=\sum_{k=1}^n\alpha_kY_k$ and $\alpha_n=(1+\|Y_n\|_{\cX'})^{-1}(1+\sigma_\cC(Y_n,1))^{-1}2^{-n}>0$ for every $n\in\N$. Since $\cX'$ is complete with respect to its norm topology, we have $S_n\to Z$ for a suitable $Z\in\cX'$ with respect to said topology. Hence, by our standing assumptions, we also have $S_n\to Z$ with respect to the reference topology $\sigma(\cX',\cX)$. To conclude the proof, note that $\sum_{k=1}^n\alpha_k\to r$ for some $r>0$ and
\[
\sigma_\cC(Z,r) \leq \liminf_{n\to\infty}\sum_{k=1}^n\alpha_k\sigma_\cC(Y_k,1) < \infty
\]
by lower semicontinuity and sublinearity of $\sigma_\cC$. This yields $(Z,r)\in\barr(\cC)$. The desired statement follows by setting $\lambda_n=\frac{\alpha_n}{r}>0$ for every $n\in\N$ and $Y=\frac{Z}{r}\in\cD$.
\end{proof}

Establishing the countable separation property is more challenging and requires an additional assumption, namely the absence of scalable good deals. In the next proposition we state a useful equivalent condition for this to hold in the case of a pointed conic acceptance set. We show that, in this case, there are no scalable good deals if and only if every good deal that is diminished by arbitrary multiples of any given (nonzero) acceptable payoff ceases to be a good deal. If the acceptance set is the standard positive cone, this condition corresponds to the ``no scalable arbitrage'' condition in Pennanen (2011a).

\begin{proposition}
\label{prop: on scalable deals}
Let $\cA$ be a pointed cone. Then, there exists no scalable good deal if and only if for every nonzero $X\in\cA\cap\cX$ there is $\lambda>0$ such that $(\cA+\lambda X)\cap\{Z\in\cM \,; \ \pi(Z)\leq0\}=\emptyset$.
\end{proposition}
\begin{proof}
To prove the ``if'' implication, assume that for every nonzero $X\in\cA\cap\cX$ there is $\lambda>0$ such that $(\cA+\lambda X)\cap\{Z\in\cM \,; \ \pi(Z)\leq0\}=\emptyset$. Take a nonzero $X\in\cA\cap\cX$. By assumption, we find $\lambda>0$ such that $\lambda X\notin\{Z\in\cM \,; \ \pi(Z)\leq0\}$. This implies that $X\notin\{Z\in\cM^\infty \,; \ \pi^\infty(Z)\leq0\}$. In particular, for every $X\in\cA^\infty\cap\cM^\infty$ such that $\pi^\infty(X)\leq0$ we must have $X\in\cS\subset\cX$ and, hence, $X=0$. This shows that no scalable good deal can exist. To prove the ``only if'' implication, assume conversely that no scalable good deal exists. First, we claim that $\{Z\in\cA\cap\cM \,; \ \pi(Z)\leq0\}$ is bounded. If this is not the case, for every $n\in\N$ we find $Z_n\in\cA\cap\cM$ such that $\pi(Z_n)\leq0$ and $\|Z_n\|\geq n$. As the unit sphere in $\cS$ is compact, there exists a nonzero $Z\in\cS$ such that $\frac{Z_n}{\|Z_n\|}\to Z$. Note that $Z\in\cA^\infty\cap\cM^\infty$ by~\eqref{eq: recession cones 1}. Note also that the lower semicontinuity and convexity of $\pi$ yield
\[
\pi(Z) \leq \liminf_{n\to\infty}\pi\left(\frac{Z_n}{\|Z_n\|}\right) \leq \liminf_{n\to\infty}\frac{\pi(Z_n)}{\|Z_n\|} \leq 0.
\]
This shows that $Z$ is a scalable good deal, contradicting our assumption. Hence, $\{Z\in\cA\cap\cM \,; \ \pi(Z)\leq0\}$ is bounded. Now, suppose we find a nonzero $X\in\cA\cap\cX$ such that for every $\lambda>0$ there exists $Z_\lambda\in\cM$ with $\pi(Z_\lambda)\leq0$ and $Z_\lambda-\lambda X\in\cA$. In particular, $Z_\lambda\in\cA$ and $\frac{Z_\lambda}{\lambda}\in\cA+X$ for every $\lambda>0$. As $(\cA+X)\cap\cS$ is closed and does not contain the zero payoff, the norm $\|\cdot\|$ must be bounded from below by a suitable $\varepsilon>0$ on the set $(\cA+X)\cap\cS$. In particular, $\frac{\|Z_\lambda\|}{\lambda}\geq\varepsilon$ for every $\lambda>0$. This implies that $\{Z_\lambda \,; \ \lambda>0\}$ is unbounded. However, this is against what we have prove above, showing that for every nonzero $X\in\cA\cap\cX$ there must be $\lambda>0$ such that $(\cA+\lambda X)\cap\{Z\in\cM \,; \ \pi(Z)\leq0\}=\emptyset$.
\end{proof}

We are finally in a position to state
sufficient conditions for the existence of strictly consistent price deflators. As a first step, we provide two sets of sufficient conditions for the existence of consistent price deflators that are strictly positive. This is achieved by proving the countable separation property for $\cL=\cX_+$ and $\cL'=\cD$. In order to move from strict positivity to strict consistency, we need an additional assumption on the dual space $\cX'$, namely the separability of its norm predual. In this case, we are able to establish the countable separation property for $\cL=\cA\cap\cX$ and $\cL'=\cD$. We refer to the accompanying remark for a detailed discussion about the proof strategy and the separability assumption.

\begin{theorem}
\label{theo: dual ftap A convex}
Assume that one of the following holds:
\begin{enumerate}
  \item[(i)] $\cA$ is a pointed cone and there exists no scalable good deal.
  \item[(ii)] $\cK(\cA)$ is pointed and there exists no scalable good deal with respect to $\cK(\cA)$.
\end{enumerate}
Then, there exists a strictly positive consistent price deflator $D$ in $\cX'$. If, in addition, the norm predual of $\cX'$ is separable with respect to its norm topology, then $D$ can be taken to be strictly consistent.
\end{theorem}
\begin{proof}
It follows from Proposition \ref{prop: conified A} that $\cK(\cA)$ is a conic acceptance set such that $\cK(\cA)\cap\cX$ is closed and coincides with $\cl(\cone(\cA)\cap\cX)$. Note that every price deflator $D$ that is (strictly) consistent with $\cK(\cA)$ is also (strictly) consistent with $\cA$. As a result, it suffices to prove the stated claims under condition (i). Hence, assume that $\cA$ is a pointed cone and there exists no scalable good deal.

\smallskip

We first show that we can always find a strictly positive consistent price deflator in $\cX'$. To this effect, we apply Theorem \ref{theo: kreps yan in our setting} to $\cL=\cX_+$ and $\cL'=\cD$, in which case $\cL'\cap(-\barr(\cone(\cL)))=\cD$ by Lemma~\ref{lem: elementary properties C}. In view of this result and of Proposition~\ref{prop: completeness property}, to establish our claim it suffices to exhibit a sequence $(Y_n)\subset\cD$ of price deflators such that
\begin{equation}
\label{eq: countable separation}
\mbox{for every nonzero $X\in\cX_+$ there exists $n\in\N$ such that $\E[XY_n]>0$}.
\end{equation}
By Proposition~\ref{prop: on scalable deals}, for every nonzero $X\in\cX_+$ there exists $\lambda>0$ such that $(\lambda X,0)\notin\cC$. Since $\cC$ is closed and $(0,n)\notin\cC$ for some $n\in\N$ by Lemma~\ref{lem: closedness C}, we can use the representation of (the closure of) $\cC$ in Lemma~\ref{lem: elementary properties C} to find an element $Y_X\in\cD$ such that $\E[\lambda XY_X]>\sigma_\cC(Y_X,1)\geq0$. Equivalently, we have that
\begin{equation}
\label{eq: countable separation preliminary}
\mbox{for every nonzero $X\in\cX_+$ there exists $Y_X\in\cD$ such that $\E[XY_X]>0$}.
\end{equation}
To establish \eqref{eq: countable separation}, we start by showing that the family $\cG = \{\{Y>0\} \,; \ Y\in\cD\}$ is nonempty and closed under countable unions. That $\cG$ is nonempty follows from~\eqref{eq: countable separation preliminary}. To show that $\cG$ is closed under countable unions, take an arbitrary sequence $(Y_n)\subset\cD\setminus\{0\}$. By Proposition~\ref{prop: completeness property}, we find a sequence $(\lambda_n)\subset(0,\infty)$ and an element $Y\in\cD$ such that $S_n=\sum_{k=1}^n\lambda_kY_k\to Y$. It is easy to see that
\begin{equation}
\label{eq: exhaustion 1}
\{Y>0\} = \bigcup_{n\in\N}\{Y_n>0\} \ \ \ \mbox{$\probp$-almost surely}.
\end{equation}
Indeed, consider first the event $E=\{Y>0\}\cap\bigcap_{n\in\N}\{Y_n=0\}$. We must have $\probp(E)=0$ for otherwise
\[
0 < \E[\one_EY] = \lim_{n\to\infty}\E[\one_ES_n] = 0.
\]
As a result, the inclusion ``$\subset$'' in~\eqref{eq: exhaustion 1} must hold. Next, we claim that $\probp(Y\geq S_n)=1$ for every $n\in\N$. If not, we find $k\in\N$ and $\varepsilon>0$ such that the event $E=\{Y\leq S_k-\varepsilon\}$ satisfies
\[
0 < \varepsilon\probp(E) \leq \E[\one_E(S_k-Y)] \leq \lim_{n\to\infty}\E[\one_E(S_n-Y)] = 0.
\]
This delivers the inclusion ``$\supset$'' in~\eqref{eq: exhaustion 1} and shows that $\cG$ is closed under countable unions as desired. Now, set $s=\sup\{\probp(E) \,; \ E\in\cG\}$. Take any sequence $(Y_n)\subset\cD$ such that $\probp(Y_n>0)\uparrow s$. By closedness under countable unions, there must exist $Y^\ast\in\cD$ such that $\{Y^\ast>0\}=\bigcup_{n\in\N}\{Y_n>0\}$ $\probp$-almost surely. Take an arbitrary nonzero $X\in\cX_+$ and assume that $\E[XY_n]=0$ for every $n\in\N$. This would imply that $\E[XY^\ast]=0$ and, thus, the element $\frac{1}{2}Y^\ast+\frac{1}{2}Y_X\in\cD$ would satisfy
\[
\probp\left(\frac{1}{2}Y^\ast+\frac{1}{2}Y_X>0\right) \geq \probp(Y^\ast>0)+\probp(\{Y^\ast=0\}\cap\{Y_X>0\}) > \probp(Y^\ast>0) = s,
\]
which cannot hold. In conclusion, we must have $\E[XY_n]>0$ for some $n\in\N$, showing~\eqref{eq: countable separation}.

\smallskip

To conclude the proof, we show that there exist a strictly consistent price deflator in $\cX'$ if we additionally assume that the norm predual of $\cX'$ is separable with respect to its norm topology. To this end, we apply Theorem \ref{theo: kreps yan in our setting} to $\cL=\cA\cap\cX$ and $\cL'=\cD$, in which case $\cL'\cap(-\barr(\cone(\cL)))=\cD$ by Lemma~\ref{lem: elementary properties C}. In view of this result and of Proposition~\ref{prop: completeness property}, we are done if we exhibit a sequence $(Y_n)\subset\cD$ such that
\begin{equation}
\label{eq: countable separation A}
\mbox{for every nonzero $X\in\cA\cap\cX$ there exists $n\in\N$ such that $\E[XY_n]>0$}.
\end{equation}
By Proposition~\ref{prop: on scalable deals}, for every nonzero $X\in\cA\cap\cX$ there exists $\lambda>0$ such that $(\lambda X,0)\notin\cC$. Since $\cC$ is closed and $(0,n)\notin\cC$ for some $n\in\N$ by Lemma~\ref{lem: closedness C}, we can use the representation of (the closure of) $\cC$ in Lemma~\ref{lem: elementary properties C} to find an element $Y_X\in\cD$ such that $\E[\lambda XY_X]>\sigma_\cC(Y_X,1)\geq0$. Equivalently, we have that
\begin{equation}
\label{eq: countable separation A preliminary}
\mbox{for every nonzero $X\in\cA\cap\cX$ there exists $Y_X\in\cD$ such that $\E[XY_X]>0$}.
\end{equation}
Recall that $\cX'$ is a norm dual and denote by $\|\cdot\|_{\cX'}$ the corresponding dual norm. For every nonzero $X\in\cA\cap\cX$ consider the rescaled couple
\[
(Z_X,r_X) = \bigg(\frac{Y_X}{\|Y_X\|_{\cX'}},\frac{1}{\|Y_X\|_{\cX'}}\bigg) \in \barr(\cC).
\]
As the norm predual of $\cX'$ is separable by assumption, the unit ball in $\cX'$ is weak-star metrizable by Theorem 6.30 in Aliprantis and Border (2006). Being weak-star compact by virtue of the Banach-Alaoglu Theorem, see e.g.\ Theorem 6.21 in Aliprantis and Border (2006), the unit ball together with any of its subsets is therefore weak-star separable. In particular, this is true for $\{Z_X \,; \ X\in(\cA\cap\cX)\setminus\{0\}\}$. Since our reference topology on $\cX'$, namely $\sigma(\cX',\cX)$, was assumed to be weaker than the weak-star topology, it follows that $\{Z_X \,; \ X\in(\cA\cap\cX)\setminus\{0\}\}$ is also separable with respect to $\sigma(\cX',\cX)$. Let $\{Z_{X_n} \,; \ n\in\N\}$ be a countable dense subset. Then, for every nonzero $X\in\cA\cap\cX$, it follows immediately from~\eqref{eq: countable separation A preliminary} that we must have $\E[XY_{X_n}]>0$ for some $n\in\N$ by density. This delivers~\eqref{eq: countable separation A}.
\end{proof}

\begin{remark}
\label{rem: assumptions dual ftap}
(i) The pointedness requirement can be slightly weakened. Indeed, it suffices that $\cA\cap\cX$ and $\cK(\cA)\cap\cX$ are pointed, respectively. In view of Proposition~\ref{prop: conified A}, the latter condition is equivalent to the pointedness of $\cl(\cone(\cA))\cap\cX$. Note that, under pointedness, the absence of scalable good deals is equivalent to the generally weaker absence of strong scalable good deals. Note also that pointedness of $\cA\cap\cX$ is necessary for the existence of strictly consistent price deflators. One can verify that pointedness of $\cA\cap\cX$ is satisfied by most of the standard acceptance sets. For instance, by Proposition 5.9 in Bellini et al.\ (2021), this holds whenever $\cX$ is law invariant and $\cA$ is a law-invariant cone such that $\cA\cap\cX\neq\{X\in\cX \,; \ \E[X]\geq0\}$.

\smallskip

(ii) The separability of the norm predual of $\cX'$ is typically ensured by suitable assumptions on the underlying $\sigma$-field. For concreteness, consider the case where $\cX'=L^\infty$, which is interesting because it delivers bounded price deflators. In this case, the norm predual is $L^1$. A simple sufficient condition for separability is that $\cF$ is countably generated. A characterization of separability in the nonatomic setting can be found, e.g., in Theorem 13.16 in Aliprantis and Border (2006). It is worthwhile highlighting that separability is not required of the reference payoff space $\cX$ and may hold even if $\cX$ is not separable with respect to a pre-specified natural topology. For instance, if $\cX$ is an Orlicz space, then separability with respect to the norm topology may fail even if $\cF$ is countably generated; see, e.g., Theorem 1 in Section 3.5 in Rao and Ren (1991).

\smallskip

(iii) To establish the existence of a strictly consistent price deflator we had to ``conify'' the acceptance set $\cA$ so as to obtain another acceptance set $\cK(\cA)$ satisfying the same standing assumptions. A direct way to see that a ``conification'' is necessary is to observe that every strictly consistent price deflator is automatically strictly consistent for the acceptance set $\cK(\cA)$. This is also true for the more natural ``conified'' acceptance set $\cone(\cA)$, but the the intersection $\cone(\cA)\cap\cX$ need not be closed and, hence, our standing assumptions need not hold.

\smallskip

(iv) The proof of the existence of  {\em strictly positive} consistent price deflators builds on the exhaustion argument underpinning the classical result on equivalent probability measures in Halmos and Savage (1949). In fact, a direct application of that result provides an alternative proof of the countable separation property in~\eqref{eq: countable separation}. To see this, note that every element $Y_X\in\cD$ in~\eqref{eq: countable separation preliminary} is associated with a probability measure on $(\Omega,\cF)$ defined by $d\probp_X = \frac{Y_X}{\E_\probp[Y_X]}d\probp$. Since the family of such probability measures is dominated by $\probp$, it follows from Lemma 7 in Halmos and Savage (1949) that there exists a sequence $(X_n)\subset\cX_+\setminus\{0\}$ such that for every $E\in\cF$ we have that $\probp_{X_n}(E)=0$ for every $n\in\N$ if and only if $\probp_X(E)=0$ for every nonzero $X\in\cX_+$. For every nonzero $X\in\cX_+$ we clearly have $\probp_X(X>0)>0$ and, hence, there must exist $n\in\N$ such that $\probp_{X_n}(X>0)>0$ or, equivalently, $\E[XY_{X_n}]>0$. The countable separation property is thus fulfilled by the sequence $(Y_{X_n})$. It is worth noting that neither this argument nor the argument in the proof above can be used to ensure the existence of {\em strictly consistent} price deflators when the acceptance sets is strictly larger than the standard positive cone and, thus, contains nonpositive payoffs. This is because controlling probabilities alone is not sufficient to control the sign of expectations. To deal with strict consistency in the general case we therefore had to pursue a different strategy based on the separability of the norm predual of $\cX'$, which was inspired by the original work by Kreps (1981) and by the related work by Clark (1993) in the setting of frictionless markets.
\end{remark}

We are finally in a position to establish the announced version of the Fundamental Theorem of Asset Pricing for markets with frictions and general acceptance sets. The previous theorem is the heart of our Fundamental Theorem, which we state in the usual form of an equivalence and where, for concreteness, we focus on bounded price deflators to mimic its classical formulation. The theorem follows at once by combining Proposition~\ref{prop: density implies no good deal} and Theorem~\ref{theo: dual ftap A convex}. We split the theorem in three parts. In a first part, we focus on the situation where the acceptance set is the standard positive cone. In this case, we obtain a different proof of the one-period version of the Fundamental Theorem in markets with frictions established in Theorem 5.4 in Pennanen (2011a). As already said, the absence of scalable arbitrage opportunities corresponds to the ``no scalable arbitrage'' condition and a price deflator corresponds to a marginal price deflator in that paper. In a second and third part, we focus on conic and general acceptance sets, respectively. The corresponding versions of the Fundamental Theorem are new. We refer to the accompanying remark for a detailed embedding in the literature and to the example below for a proof of the necessity of our assumptions on the acceptance set.

\begin{theorem}[{\bf Fundamental Theorem of Asset Pricing}]
\label{theo: FTAP}
\begin{enumerate}
\item[(i)] There exists no scalable arbitrage opportunity if and only if there exists a strictly positive price deflator in $L^\infty$.
\item[(ii)] Let $L^1$ be separable with respect to its norm topology (e.g., $\cF$ is countably generated).
\begin{enumerate}
    \item[(a)] Let $\cA$ be a pointed cone. Then, there exists no scalable good deal if and only if there exists a strictly consistent price deflator in $L^\infty$.
    \item[(b)] Let $\cK(\cA)$ be pointed. If there exists no scalable good deal with respect to $\cK(\cA)$, then there exists a strictly consistent price deflator in $L^\infty$. If there exists a strictly consistent price deflator in $L^\infty$, then there exists no scalable good deal.
\end{enumerate}
\end{enumerate}
\end{theorem}

\vspace{0.01cm}

\begin{remark}
\label{rem: FTAP}
We provide a detailed comparison of our version of the Fundamental Theorem of Asset Pricing with the various versions obtained in the good deal pricing literature.

\smallskip

(i) The focus of Carr et al.\ (2001) is on one-period frictionless markets. The reference acceptance set is convex and defined in terms of finitely many test probabilities. The reference probability space is finite. In Theorem 1 the authors establish a Fundamental Theorem under the absence of a special type of good deals that is specific to the polyhedral structure of the acceptance set and that is stronger than the absence of scalable good deals. The statement is in terms of representative state pricing functions, which correspond to special (in general not strictly) consistent price deflators.

\smallskip

(ii) The focus of Jaschke and K\"{u}chler (2001) is on multi-period markets with proportional frictions admitting a frictionless asset. The reference acceptance set is assumed to be a convex cone. The reference probability space is general. In fact, the payoff space is an abstract topological vector space. In Corollary 8 the authors establish a Fundamental Theorem under the assumption of absence of good deals of second kind. In our setting, this is equivalent to the absence of payoffs $X\in\cA\cap\cM$ such that $\pi(X)<0$. The statement is in terms of consistent (not strictly consistent) price deflators. To deal with the infinite dimensionality of $\cM$, which follows from the multi-period nature of the market model, the Fundamental Theorem is stated under an additional assumption that corresponds to the closedness of $\cC$. No sufficient conditions for this are provided. It should be noted that the absence of good deals of second kind is not sufficient to ensure closedness of $\cC$ even when $\cM$ is finite dimensional. To show this, let $\Omega=\{\omega_1,\omega_2,\omega_3\}$ and assume that $\cF$ is the power set of $\Omega$ and that $\probp(\omega_1)=\probp(\omega_2)=\probp(\omega_3)=\frac{1}{3}$. We take $\cX=L^0$ and identify every element of $L^0$ with a vector of $\R^3$. Let $\cM$ coincide with $\cS=\{(x,y,z)\in\R^3 \, ; \ x=0\}$ and let $\pi:\cS\to\R$ be defined by $\pi(x,y,z)=y$. Consider the closed convex conic acceptance set
\[
\cA=\left\{(x,y,z)\in\R^3 \,; \ x^2+y^2+6xy+2\sqrt{6}xz+2\sqrt{6}yz\geq0, \ \sqrt{3}x+\sqrt{3}y+\sqrt{2}z\geq0 \right\},
\]
obtained by rotating the cone $\cA'=\{(x,y,z)\in\R^3 \, ; \ x^2+y^2\leq 3 z^2, \ z\geq0\}$ by $\pi/3$ around the direction $(-1,1,0)$. It is easy to verify that if $X\in\cA\cap\cM$, then $\pi(X)\geq0$ and, hence, there are no good deals of second kind. We show that $\cC$ is not closed. For every $n\in\N$ define $X_n=\left(1-\frac{1}{n},-1,0\right)$ and note that $(X_n,0)\in\cC$ because $Z_n=(0,0,n^2)\in\cM$ satisfies $\pi(Z_n)=0$ and $Z_n-X_n\in\cA$. Clearly, we have $(X_n,0)\to(X,0)$ with $X=(1,-1,0)$. We conclude that $\cC$ is not closed as $(X,0)\notin\cC$.

\smallskip

(iii) The focus of \v{C}ern\'{y} and Hodges (2002) is on one-period frictionless markets. The reference acceptance set is convex. The reference probability space is general. In fact, the payoff space is an abstract locally-convex topological vector space. In Theorem 2.5 the authors establish a Fundamental Theorem under the absence of good deals with respect to the ``conified'' acceptance set. The statement is expressed in terms of strictly consistent price deflators and is proved under the additional assumption that $\cX$ is an $L^p$ space for some $1<p<\infty$ and that $\cA$ is boundedly generated, i.e., is included in the cone generated by a bounded set. This condition typically fails when the underlying probability space is not finite.

\smallskip

(iv) The focus of Staum (2004) is on multi-period markets with convex frictions. The reference acceptance set is convex. The reference probability space is general. In fact, the payoff space is an abstract locally-convex topological vector space. In Theorem 6.2 the author establishes a Fundamental Theorem under the assumption that for all payoffs $X\in\cX$ and nonzero $Z\in\cX_+$
\[
\inf\{\pi(Z) \,; \ Z\in\cM, \ Z-X\in\cA\}+\inf\{\pi(Z) \,; \ Z\in\cM, \ Z-X\in\cX_+\} > 0.
\]
The link with the absence of good deals is not discussed. The statement is in terms of strictly positive (not strictly consistent) price deflators. To deal with the infinite dimensionality of $\cM$, which follows from the multi-period nature of the market model, the Fundamental Theorem is stated under the additional assumption that $\pi^+$ is lower semicontinuous. Sufficient conditions for this are provided when $\cX=L^\infty$ (with respect the standard norm topology). Unfortunately, the proof of Lemma 6.1, which is key to deriving the Fundamental Theorem, is flawed. On the one side, Zorn's Lemma is evoked to infer that a family of sets that is closed under countable unions admits a maximal element. However, this is not true as illustrated, for instance, by the family of all countable subsets of $\R$. On the other side, it is tacitly assumed that, for a generic dual pair $(\cX,\cX')$, the series $\sum_{n\in\N}2^{-n}Y_n$ converges in the topology $\sigma(\cX',\cX)$ for every choice of $(Y_n)\subset\cX'$, which cannot hold unless special assumptions are required of the pair $(\cX,\cX')$ (as those stipulated, e.g., in Assumption~\ref{standing assumption}). The underlying strategy of reproducing the exhaustion argument used in the classical proof of the Fundamental Theorem seems unlikely to work because it heavily relies on the existence of a (dominating) probability measure and, as highlighted in Remark~\ref{rem: assumptions dual ftap}, breaks down in the presence of nonpositive acceptable payoffs.

\smallskip

(v) The focus of Cherny (2008) is on one-period markets with convex frictions. The reference acceptance set is a convex cone. The reference payoff space is tailored to the chosen acceptance set by way of a duality construction, which often delivers standard $L^p$ spaces, for example, when the acceptance set is based on expected shortfall. In Theorem 3.1 the author establishes a version of the Fundamental Theorem under the absence of special good deals. In our setting, they correspond to payoffs $X\in\cM$ with $\pi(X)\leq0$ and
\[
\inf\{m\in\R \,; \ X+m\in\cA\} < 0.
\]
The statement is in terms of a special class of (not necessarily strictly positive) price deflators. The proof uses the additional assumption that the barrier cone of the acceptance set is compactly generated.

\smallskip

(vi) The focus of Madan and Cherny (2010) is on one-period frictionless markets. The reference acceptance set is induced by an acceptability index. The reference payoff space consists of suitably integrable random variables. In Theorem 1 the authors provide a version of the Fundamental Theorem under the absence of good deals. The statement is in terms of (not necessarily strictly positive) price deflators.

\smallskip

(vii) The focus of Cheridito et al.\ (2017) is on multi-period markets with general frictions admitting a frictionless asset. The reference acceptance set is also general but is required to ensure convexity of a set that, in our notation, corresponds to
\[
\{X\in\cX \,; \ \exists Z\in\cM, \ \pi(Z)\leq0, \ Z-X\in\cA\} = \{X\in\cX \,; \ (X,0)\in\cC\}.
\]
The payoff space consists of suitable regular stochastic processes. Notably, no dominating probability measure is assumed to exist. In Theorem 2.1 the authors establish a Fundamental Theorem under the absence of a suitable class of strong good deals. To deal with the infinite dimensionality of $\cM$, which follows from the multi-period nature of the market model, the Fundamental Theorem is stated under additional regularity assumptions on the market model and the acceptance set ensuring finiteness of superreplication prices of special call options. The statement is in terms of (not necessarily strictly) consistent
price deflators.
\end{remark}

\begin{example}
Let $\Omega=\{\omega_1,\omega_2\}$ and assume that $\cF$ is the power set of $\Omega$ and that $\probp$ satisfies $\probp(\omega_1)=\probp(\omega_2)=\frac{1}{2}$. In this setting, we take $\cX=\cX'=L^0$ and identify every element of $L^0$ with a vector of $\R^2$.

\smallskip

(i) Set $\cM=\R^2$ and $\pi(x,y)=\max\{x,y\}$ for every $(x,y)\in\R^2$ and define
\[
\cA = \R^2_+\cup\{(x,y)\in\R^2 \,; \ x<0, \ y\geq x^2\}.
\]
Note that $\cA$ is not a cone. Note also that no scalable good deal exists. However, there exists no strictly consistent price deflator $D=(d_1,d_2)$. Indeed, for every $\lambda>0$ we could otherwise take $X_\lambda=(-\lambda,\lambda^2)\in\cA$ and note that $\E[DX_\lambda]>0$ implies $ d_2\lambda>d_1$, which contradicts the strict positivity, hence the strict consistency, of $D$. This shows that, if we remove conicity, then the ``only if'' implication in assertion (a) in Theorem \ref{theo: FTAP} generally fails. It also shows that the converse of the second implication in assertion (b) in the same result generally fails as well.

\smallskip

(ii) Set $\cM=\R^2$ and $\pi(x,y)=x+y$ for every $(x,y)\in\R^2$ and define
\[
\cA = \R^2_+\cup\{(x,y)\in\R^2 \,; \ x<0, \ y\geq e^{-x}-1\}.
\]
Note that $\cA$ is not a cone and $\cK(\cA)=\R^2_+\cup\{(x,y)\in\R^2 \,; \ x<0, \ y\geq -x\}$ is pointed. Note also that $D=(2,2)$ is a (in fact, the only) strictly consistent price deflator. However, $X=(-1,1)\in\cK(\cA)\cap\cM$ satisfies $\pi(X)=0$ and is therefore a scalable good deal with respect to $\cK(\cA)$. This shows that the converse of the first implication in assertion (b) in Theorem \ref{theo: FTAP} generally fails.
\end{example}


\subsection{Superreplication duality}

In this section, we first derive a dual representation of superreplication prices based on consistent price deflators under the assumption that the market is free of strong scalable good deals. We refer to Corollary 8 in Jaschke and K\"{u}chler (2001), Theorem 4.1 in Staum (2004), and Theorem 2.1 in Cheridito et al.\ (2017) for similar representations under the assumption of absence of good deals. We also refer to Proposition 3.9 in Frittelli and Scandolo (2006) for a similar representation in a risk measure setting. These representations were obtained under the assumption of lower semicontinuity of $\pi^+$. As mentioned in the proof of Proposition~\ref{prop: direct FTAP}, a sufficient condition for this to hold is precisely the absence of strong scalable good deals. In a second step, we improve the dual representation by replacing consistency with strict consistency. In a frictionless setting where the acceptance set is the standard positive cone, this is equivalent to moving from price deflators to strictly positive price deflators (equivalently, from martingale measures to equivalent martingale measures). This sharper result therefore extends the classical result on superreplication duality to markets with frictions and general acceptance sets.

\begin{theorem}[{\bf Superreplication duality}]
\label{theo: superhedging theorem}
The following statements hold:
\begin{enumerate}
  \item[(i)] If there exists no strong scalable good deal, then for every $X\in\cX$
\[
\pi^+(X) = \sup_{D\in\cD}\{\E[DX]-\gamma_{\pi,\cM}(D)+\gamma_\cA(D)\}.
\]
  \item[(ii)] If there exists no scalable good deal and if either $\cA=L^0_+$ or $\cA$ is a pointed cone and the norm predual of $\cX'$ is separable with respect to its norm topology, then for every $X\in\cX$
\begin{equation}
\label{eq: dual representation superreplication}
\pi^+(X) = \sup_{D\in\cD_{str}}\{\E[DX]-\gamma_{\pi,\cM}(D)\}.
\end{equation}
\end{enumerate}
\end{theorem}
\begin{proof}
Assume the market is free of strong scalable good deals. It follows from Lemma~\ref{lem: closedness C} that $\cC$ is closed and $(0,n)\notin\cC$ for some $n\in\N$. Now, take an arbitrary $X\in\cX$. Combining the representation of $\pi^+(X)$ in Lemma~\ref{lem: superreplication C} with the representation of (the closure of) $\cC$ obtained in Lemma~\ref{lem: elementary properties C}, we infer that
\begin{align*}
\pi^+(X)
&=
\inf\{m\in\R \,; \ \E[DX]-m-\gamma_{\pi,\cM}(D)+\gamma_\cA(D)\leq0, \ \forall D\in\cD\} \\
&=
\inf\{m\in\R \,; \ m\geq\E[DX]-\gamma_{\pi,\cM}(D)+\gamma_\cA(D), \ \forall D\in\cD\} \\
&=
\sup\{\E[DX]-\gamma_{\pi,\cM}(D)+\gamma_\cA(D) \,; \ D\in\cD\}.
\end{align*}
This proves (i). Now, let the assumptions in point (ii) hold. It follows from Theorem~\ref{theo: dual ftap A convex} that $\cD_{str}$ is nonempty. Moreover, by Lemma~\ref{lem: closedness C}, $\cC$ is closed and $(0,n)\notin\cC$ for some $n\in\N$. We claim that the representation in Lemma~\ref{lem: elementary properties C} for (the closure of) $\cC$ can be rewritten as
\begin{equation}
    \label{eq: dual repc C Dstr}
    \cC = \bigcap_{Y\in\cD_{str}}\{(X,m)\in\cX\times\R \,; \ \E[XY]+m\leq\gamma_{\pi,\cM}(Y)\}.
\end{equation}
Note that $\gamma_\cA(Y)=0$ for every $Y\in\cD$ by conicity of $\cA$. Clearly, we only need to establish the inclusion ``$\supset$''. To this end, take any $(X,m)\in\cX\times\R$ such that $\E[XY]+m\leq\gamma_{\pi,\cM}(Y)$ for every $Y\in\cD_{str}$. Fix $Y^\ast\in\cD_{str}$ and take any $Y\in\cD$. For every $\lambda\in(0,1)$ we have $\lambda Y^\ast+(1-\lambda)Y\in\cD_{str}$ so that
\begin{align*}
\lambda(\E[XY^\ast]+m)+(1-\lambda)(\E[XY]+m)
&=
\E[X(\lambda Y^\ast+(1-\lambda)Y)]+m \\
&\leq
\gamma_{\pi,\cM}(\lambda Y^\ast+(1-\lambda)Y) \\
&\leq
\lambda\gamma_{\pi,\cM}(Y^\ast)+(1-\lambda)\gamma_{\pi,\cM}(Y).
\end{align*}
Letting $\lambda\downarrow0$ delivers $\E[XY]+m\leq\gamma_{\pi,\cM}(Y)$ and shows the desired inclusion. Now, take any payoff $X\in\cX$. It follows from Lemma~\ref{lem: superreplication C} and \eqref{eq: dual repc C Dstr} that
\begin{align*}
\pi^+(X)
&=
\inf\{m\in\R \,; \ \E[DX]-m\leq\gamma_{\pi,\cM}(D), \ \forall D\in\cD_{str}\} \\
&=
\inf\{m\in\R \,; \ m\geq\E[DX]-\gamma_{\pi,\cM}(D), \ \forall D\in\cD_{str}\} \\
&=
\sup\{\E[DX]-\gamma_{\pi,\cM}(D) \,; \ D\in\cD_{str}\}.
\end{align*}
This establishes (ii) and concludes the proof.
\end{proof}


\subsection{Dual characterization of market-consistent prices}

The Fundamental Theorem also allows to derive our desired dual characterization of market-consistent prices with acceptable risk, which extends the classical characterization of arbitrage-free prices in terms of strictly positive price deflators. We complement this by showing that, contrary to the standard frictionless setting, for an attainable payoff with market-consistent superreplication price the supremum in the dual representation of the corresponding superreplication price need not be attained. Interestingly enough, this implies that a dual characterization of market-consistent prices for replicable payoffs in terms of strictly consistent price deflators is not always possible. The accompanying proposition shows a situation where the dual characterization holds also for replicable payoffs.

\begin{proposition}[{\bf Dual characterization of market-consistent prices}]
\label{theo: dual MCP}
If there exists no scalable good deal and if either $\cA=L^0_+$ or $\cA$ is a pointed cone and the norm predual of $\cX'$ is separable with respect to its norm topology, then the following statements hold:
\begin{enumerate}
  \item[(i)] If $\pi^+(X)\in\MCP(X)$ and the supremum in \eqref{eq: dual representation superreplication} is attained or if $\pi^+(X)\notin\MCP(X)$, then
\begin{equation}
\label{eq: dual representation MCP}
\MCP(X) = \{p\in\R \,; \ \exists D\in\cD_{str} \,:\, p\leq\E[DX]-\gamma_{\pi,\cM}(D)\}.
\end{equation}
  \item[(ii)] If $\pi^+(X)\in\MCP(X)$ and the supremum in \eqref{eq: dual representation superreplication} is not attained, then the strict inclusion ``$\supset$'' holds in \eqref{eq: dual representation MCP}. This can occur even if both $\pi$ and $\cM$ are conic and there exists no good deal.
\end{enumerate}
\end{proposition}
\begin{proof}
It follows from Theorem~\ref{theo: dual ftap A convex} that $\cD_{str}$ is nonempty. First, we show the inclusion ``$\supset$'' in \eqref{eq: dual representation MCP}. Let $D\in\cD_{str}$. Note that for every attainable payoff $Z\in\cM$ such that $Z-X\in\cA\setminus\{0\}$ we have
\[
\pi(Z)
\geq
\E[DZ]-\gamma_{\pi,\cM}(D)
=
\E[D(Z-X)]+\E[DX]-\gamma_{\pi,\cM}(D)
>
\E[DX]-\gamma_{\pi,\cM}(D)
\]
by strict consistency. Note also that $\E[DX]-\gamma_{\pi,\cM}(D)\leq\pi(X)$ in the case that $X\in\cM$. This shows that $\E[DX]-\gamma_{\pi,\cM}(D)$ is a market-consistent price for $X$ and yields the desired inclusion. Now, recall from Proposition~\ref{prop: interval MCP} that $\pi^+(X)$ is the supremum of the set $\MCP(X)$. If $\pi^+(X)$ belongs to $\MCP(X)$, then the inclusion ``$\supset$'' in \eqref{eq: dual representation MCP} is an equality if and only if the supremum in \eqref{eq: dual representation superreplication} is attained. We refer to Example~\ref{ex: MCP are not represented via pricing densities} for a concrete situation where the latter condition fails even if both $\pi$ and $\cM$ are conic and the market admits no good deals. Finally, assume that $\pi^+(X)$ does not belong to $\MCP(X)$. To complete the proof we only have to show the inclusion ``$\subset$'' in \eqref{eq: dual representation MCP}. To this effect, take an arbitrary market-consistent price $p\in\MCP(X)$ and note that we must have $p<\pi^+(X)$. Hence, it follows from the representation \eqref{eq: dual representation superreplication} that $p<\E[DX]-\gamma_{\pi,\cM}(D)$ for a suitable $D\in\cD_{str}$. This concludes the proof.
\end{proof}

\begin{example}
\label{ex: MCP are not represented via pricing densities}
Let $\Omega=\{\omega_1,\omega_2\}$ and assume that $\cF$ is the power set of $\Omega$ and that $\probp$ is specified by $\probp(\omega_1)=\probp(\omega_2)=\frac{1}{2}$. In this simple setting, we take $\cX=\cX'=L^0$ and identify every element of $L^0$ with a vector of $\R^2$. Take $\cA=\R^2_+$, $\cS=\R^2$ and $\cM=\{(x,y)\in\R^2 \, ; \ 0\leq y\leq-x\}$. Define
\[
\pi(x,y)=
\begin{cases}
-\sqrt{x^2+xy} & \mbox{if} \ (x,y)\in\cM\\
\infty & \mbox{otherwise}
\end{cases},
\]
which is convex because it is continuous on $\cM$ and its Hessian matrix in the interior of $\cM$ 
has nonnegative eigenvalues, namely $0$ and $\frac{1}{4}(x^2+y^2)(x^2+xy)^{-3/2}$. Both $\cA$ and $\cM$ are cones and $\pi$ is conic. Moreover, there exists no good deal. A direct inspection shows that strictly consistent price deflators $D\in\cX'$ exist (for instance, take $D=(2,1)$) and satisfy $\gamma_{\pi,\cM}(D)=0$ by conicity. Now, set $X=(-1,1)\in\cM$. We have that $\pi^+(X)=\pi(X)=0$ since $(\cA+X)\cap\cM=\{X\}$. This also yields $0\in\MCP(X)$ by Proposition \ref{prop: characterization mcp superreplication}. We show that there is no $D=(d_1,d_2)\in\cD_{str}$ such that $\E[DX]=0$. Indeed, we would otherwise have $d_1=d_2$ and taking $Z_\lambda=(-1,\lambda)\in\cM$ for $\lambda\in(0,1)$ would deliver
\[
\sup_{0<\lambda<1}\{\E[DZ_\lambda]-\pi(Z_\lambda)\}\leq0 \ \implies \ d_1\geq\sup_{0<\lambda<1}\frac{2}{\sqrt{1-\lambda}}=\infty.
\]
As a result, the supremum in~\eqref{eq: dual representation superreplication} is not attained.
\end{example}

\begin{proposition}
If $\cA$ is a cone and there exists a strictly consistent price deflator $D\in\cX'$ such that $\gamma_{\pi,\cM}(D)=0$, then for every $X\in\cX$ such that $\pi^+(X)\in\MCP(X)$ and such that $X\in\cM^\infty\cap(-\cM^\infty)$ and $\pi$ is linear on $\Span(X)$ we have
\[
\MCP(X) = \{p\in\R \,; \ \exists D\in\cD_{str} \,:\, p\leq\E[DX]\}.
\]
\end{proposition}
\begin{proof}
It follows from Proposition \ref{prop: density implies no good deal} that the market has no scalable good deals. Now, take a payoff $X\in\cX$ such that $\pi^+(X)\in\MCP(X)$ and assume that $X\in\cM^\infty\cap(-\cM^\infty)$ and $\pi$ is linear on $\Span(X)$. By Proposition~\ref{theo: characterization mcp superreplication} we have $\pi^+(X)=\pi(X)$. Moreover, by Proposition~\ref{prop: pricing density}, we know that $\pi(X)=\E[DX]$. Hence the supremum in \eqref{eq: dual representation superreplication} is attained and the thesis follows from Proposition \ref{theo: dual MCP}.
\end{proof}


The next example shows that conicity is necessary for both Theorem~\ref{theo: superhedging theorem} and Proposition~\ref{theo: dual MCP} to hold.

\begin{example}
Let $\Omega=\{\omega_1,\omega_2\}$ and assume that $\cF$ is the power set of $\Omega$ and that $\probp$ is specified by $\probp(\omega_1)=\probp(\omega_2)=\frac{1}{2}$. In this simple setting, we take $\cX=\cX'=L^0$ and identify every element of $L^0$ with a vector of $\R^2$. Define $\pi(x,y)=\max\{x,x+y\}$ for every $(x,y)\in\R^2$ and set
\[
\cM=\{(x,y)\in\R^2 \, ; \ y\geq0\}, \ \ \ \ \cA = \{(x,y)\in\R^2 \,; \ y\geq\max\{-2x,0\}, \ x\geq-1\}.
\]
Note that $\pi$ and $\cM$ are both conic while $\cA$ is not. Note also that there exists no good deal. It is not difficult to verify that strictly consistent price deflators exist. Indeed, for a strictly-positive $D=(d_1,d_2)$
\[
\begin{cases}
\sup\{\E[DX]-\pi(X) \,; \ X\in\cM\}<\infty \\
\mbox{$\E[DX]>0$ for every nonzero $X\in\cA$}
\end{cases}
 \ \iff \
\begin{cases}
d_1=2 \\
1<d_2\leq 2
\end{cases}
.
\]
Set $X=(2,-4)\in\cX$. Since $(\cA+X)\cap\cM=\{(x,y)\in\R^2 \,; \ x\geq1, \ y\geq0\}$, we see that $\pi^+(X)=\pi(1,0)=1$. As $X$ does not belong to $\cM$, we have $\MCP(X)=(-\infty,1)$ by Proposition~\ref{theo: characterization mcp superreplication}. Both \eqref{eq: dual representation superreplication} and \eqref{eq: dual representation MCP} fail, since for every strictly consistent price deflator $D=(d_1,d_2)$ we have $\gamma_{\pi,\cM}(D)=0$ by conicity and
\[
\sup_{D\in\cD_{str}}\{\E[DX]-\gamma_{\pi,\cM}(D)\} = \sup_{1<d_2\leq 2}\{2-2d_2\} = 0.
\]
\end{example}

\section{Conclusions}

We established a version of the Fundamental Theorem of Asset Pricing in incomplete markets with frictions where agents use general acceptance sets to define good deals based on their individual preferences. The basic result states that the absence of scalable good deals is equivalent to the existence of strictly consistent price deflators. This extends and sharpens the existing versions of the Fundamental Theorem in the good deal pricing literature and allows to derive the appropriate version of superreplication duality. Even though our focus in on one-period models, we had to cope with technical challenges as the standard techniques used in arbitrage pricing (changes of numeraire, exhaustion arguments) break down in the presence of general acceptance sets. The new concepts and strategies developed in the paper are meant to be the building blocks for the construction of a complete multi-period theory of good deal pricing.


\appendix

\section{Appendix}

We use the convention $\infty-\infty=-\infty$ and $0\cdot\infty=0$. A set $\cC$ in a (topological) vector space $\cX$ is {\em pointed} if $\cC\cap(-\cC)=\{0\}$, {\em convex} if $\lambda\cC+(1-\lambda)\cC\subset\cC$ for every $\lambda\in(0,1)$ and {\em conic} (or a {\em cone}) if $\lambda\cC\subset\cC$ for every $\lambda\in[0,\infty)$. The smallest linear space that contains $\cC$ is denoted by $\Span(\cC)$. Similarly, the smallest cone that contains $\cC$ is denoted by $\cone(\cC)$. If $\cC$ is convex and $0\in\cC$, its {\em recession cone} is
\[
\cC^\infty := \bigcap_{\lambda\in(0,\infty)}\lambda\cC.
\]
Note that $\cC^\infty$ is the largest convex cone contained in $\cC$. If $\cC$ is additionally closed, then $\cC^\infty$ is also closed. In this case, we can equivalently express $\cC^\infty$ as
\begin{equation}
\label{eq: recession cones 1}
\cC^\infty = \{X\in\cX \,; \ \mbox{$\exists$ nets $(X_\alpha)\subset\cC$ and $(\lambda_\alpha)\subset\R_+$} \,:\, \lambda_\alpha\downarrow0, \ \mbox{$\lambda_\alpha X_\alpha\to X$}\}=\{X\in\cX \, ; \ X+\cC\subset\cC\}.
\end{equation}
A functional $\varphi:\cX\to(-\infty,\infty]$ is {\em convex} if $\varphi(\lambda X+(1-\lambda)Y)\leq\lambda\varphi(X)+(1-\lambda)\varphi(Y)$ holds for all $X,Y\in\cX$ and $\lambda\in(0,1)$, {\em conic} if $\varphi(\lambda X)=\lambda\varphi(X)$ holds for all $X\in\cX$ and $\lambda\in[0,\infty)$, {\em sublinear} if $\varphi$ is simultaneously convex and conic, {\em lower semicontinuous} if for every net $(X_\alpha)\subset\cX$ and every $X\in\cX$ such that $X_\alpha\to X$, we have $\varphi(X)\leq\liminf_\alpha\varphi(X_\alpha)$. This is equivalent to $\{X\in\cX \,; \ \varphi(X)\leq m\}$ being closed for every $m\in\R$. If $\varphi$ is convex and $\varphi(0)=0$, its {\em recession functional} $\varphi^\infty:\cX\to[-\infty,\infty]$ is
\[
\varphi^\infty(X) := \sup_{\lambda>0}\frac{\varphi(\lambda X)}{\lambda}.
\]
It is the smallest sublinear functional dominating $\varphi$. If $\varphi$ is lower semicontinuous, then also $\varphi^\infty$ is and for every $m\in\R$ we have
\begin{equation}
\label{eq: recession cones 2}
\{X\in\cX \,; \ \varphi(X)\leq m\}^\infty = \{X\in\cX \,; \ \varphi^\infty(X)\leq0\}.
\end{equation}
Denote by $\cY$ the topological dual of $\cX$ and by $\sigma(\cX,\cY)$ the weakest linear topology on $\cX$ such that the map $\langle \cdot,Y\rangle$ is continuous for every $Y\in\cY$. The {\em (upper) support functional} of a (nonempty) set $\cC\subset\cX$ is the map $\sigma_\cC:\cY\to(-\infty,\infty]$ defined by $\sigma_\cC(Y) := \sup_{X\in\cC}\langle X,Y\rangle$. Note that $\sigma_\cC$ is sublinear and $\sigma(\cY,\cX)$-lower semicontinuous. The effective domain of $\sigma_\cC$, $\barr(\cC) := \{Y\in\cY \,; \ \sigma_\cC(Y)<\infty\}$, is called the {\em barrier cone} of $\cC$. Note that $\barr(\cC)$ is a convex cone and, unless $\cC$ is a cone, may fail to be $\sigma(\cY,\cX)$-closed. If $\cC$ is a cone, then $\barr(\cC) = \{Y\in\cY \,; \ \langle X,Y\rangle\leq0, \ \forall X\in\cC\}$. 


\end{document}